\documentclass[runningheads,envcountsect,envcountsame,11pt]{llncs}
\usepackage{amsmath,amssymb,euscript}
\usepackage{framed}
\usepackage{balance,complexity,pgfplots,tikz,tikz-3dplot}
\usetikzlibrary{shapes}
\usepackage[ruled,linesnumbered,vlined]{algorithm2e}

\usepackage{hyperref}

\usepackage{todonotes}

\def\ES#1{{\EuScript#1}}

\def\lseg#1#2{\overline{#1#2}}
\def\upseg{\mathop{\mbox{\sl ups}}}
\def\uptan{\mathop{\mbox{\sl upt}}}
\def\lowseg{\mathop{\mbox{\sl los}}}
\def\lowtan{\mathop{\mbox{\sl lot}}}
\let\lotan\lowtan

\long\def\useappendix#1{}
\long\def\nouseappendix#1{#1}
\useappendix{\long\def\nouseappendix#1{}}
\useappendix{\usepackage[appendix=append]{apxproof}}%=strip
\nouseappendix{\usepackage[appendix=inline]{apxproof}}
\nouseappendix{%
\usepackage[margin=2.6cm]{geometry}
}

\begin{document}
	
\title{On conflict-free chromatic guarding \mbox{of simple polygons}}

\author{Onur \c{C}a\u{g}{\i}r{\i}c{\i}\,\inst{1}%
\thanks{Supported by the {Czech Science Foundation}, project no.~{20-04567S}.} \and
	Subir Kumar Ghosh\,\inst{2}%
\thanks{Supported by {SERB}, Government of India through a grant under {MATRICS}.} \and
\mbox{Petr Hlin\v en\'y}\,\inst{1}${}^\star$ \and
	Bodhayan Roy\,\inst{3}%
\thanks{A part of the research was done while B. Roy was
	affiliated with the Faculty of Informatics of Masaryk University.} }

\authorrunning{O.\c{C}a\u{g}{\i}r{\i}c{\i}, S.K. Ghosh, P. Hlin\v en\'y and B. Roy}

\institute{Faculty of Informatics of Masaryk University, Brno, Czech Republic
\\\email{onur@mail.muni.cz, hlineny@fi.muni.cz} 
\smallskip \and
Ramakrishna Mission Vivekananda University, Kolkata, India
\\\email{ghosh@tifr.res.in}
\smallskip \and
Indian Institute of Technology Kharagpur, India
\\\email{broy@maths.iitkgp.ac.in}}

	\maketitle

 \newtheoremrep{proposition}[theorem]{Proposition}
 \newtheoremrep{lemma}[theorem]{Lemma}
 \newtheoremrep{corollary}[theorem]{Corollary}
 \newtheoremrep{theorem2}[theorem]{Theorem}

	\begin{abstract}
        We study the problem of colouring the vertices of a polygon, such
        that every viewer in it can see a unique colour.  The goal is to minimise the
        number of colours used.  This is also known as the conflict-free 
        chromatic guarding problem with vertex guards, and is
        motivated, e.g., by the problem of radio frequency assignment
        to sensors placed at the polygon vertices.

        We first study the scenario in which viewers can be all points of the polygon
        (such as a mobile robot which moves in the interior of the polygon).  
        We efficiently solve the related problem of minimising the number of guards
        and approximate (up to only an additive error) 
        the number of colours required in the special case of polygons called funnels.
        Then we give an upper bound of $O(\log^2 n)$ colours on
        $n$-vertex weak visibility polygons, by decomposing the problem into sub-funnels.
        This bound generalises to all simple polygons. 

        We briefly look also at the second scenario, in which the viewers are only the
        vertices of the polygon.
        We show a lower bound of $3$ colours in the general case of
        simple polygons and conjecture that this is tight.
        We also prove that already deciding whether $1$ or $2$ colours are enough
        is \NP-complete. % problem for simple polygons.

	\end{abstract}

\title{On conflict-free chromatic guarding of simple polygons
	\thanks{O.~\c{C}a\u{g}{\i}r{\i}c{\i} and P.~Hlin\v{e}n\'y
	have been supported by the {Czech Science Foundation}, project no.~{20-04567S}.
	S. K. Ghosh has been supported by {SERB}, Government of India through a grant under {MATRICS}.
	A significant part of the research was done while B. Roy was affiliated
	with the Faculty of Informatics	of Masaryk University.
}
}

\author{Onur \c{C}a\u{g}{\i}r{\i}c{\i} \and
        Subir Kumar Ghosh \and
        Petr Hlin\v{e}n\'{y} \and
        Bodhayan Roy
}

\institute{
			O. \c{C}a\u{g}{\i}r{\i}c{\i} \at
			Faculty of Informatics, Masaryk University, Brno - Czech Republic\\
			\email{onur@mail.muni.cz}
			        \and
			S. K. Ghosh \at
			Ramakrishna Mission Vivekananda University, Kolkata, India\\
			\email{ghosh@tifr.res.in}
				\and
			P. Hlin\v{e}n\'{y} \at
			Faculty of Informatics, Masaryk University, Brno - Czech Republic\\
			\email{hlineny@fi.muni.cz}
				\and
			B. Roy \at
			Indian Institute of Technology Kharagpur, India\\
			\email{broy@maths.iitkgp.ac.in}
}

\section{Introduction}
	The \emph{guarding} of a polygon considers placing ``guards'' into the polygon, in a way that the guards collectively can see the whole polygon.
	It is usually assumed that a guard can see any point unless there is an obstacle or a wall between the guard and that point.
	One of the best known problems  in computational geometry, the \emph{art gallery
	problem}, is essentially a guarding problem \cite{bcko-cgaa-08,o-agta-87}.
	The problem is to find the minimum number of guards to guard an art gallery, which is
	modelled by an $n$-vertex polygon.
	This problem  was shown to be NP-hard by Lee and Lin \cite{ll-ccagp-86} and more recently $\exists \mathbb{R}$-complete by Abrahamsen et al. \cite{art-gal-etr}. 
	The Art Gallery Theorem, proved by Chv{\'a}tal, shows that $ \lfloor n/3 \rfloor$ guards are sufficient and sometimes necessary to guard a simple polygon \cite{c-actpg-75}. 
	
	The guard minimisation problem has been studied under many
	constraints; such as the placement of guards being restricted to the
	polygonal perimeter or vertices \cite{loglogn-artgal}, the viewers
	being restricted to vertices, the polygon being terrains
	\cite{saurabh-artgal,ben-moshe-terrain,terrain-gal}, weakly visible
	from an edge \cite{weakvis-guard}, with holes or orthogonal
	\cite{eidenbenz_inapprox,katz-orth-artgal,terrain-np-complete}, with
	respect to parameterization \cite{para-artgal} and to approximability
	\cite{katz-wvp}.
	
	For most of these cases the problem remains hard, but interesting approximation algorithms have also been provided \cite{g-apatpp-2010,approx-art-gal-bonnet}.
	
	In addition to above mentioned versions of art gallery problem (or rather polygon guarding problem), some problems consider not the number of the guards, but the number of colours that are assigned to the guards.
	The colours, depending on the scope, determine the types of the guards.
	If any observer in the polygon sees at least one guard of a different type
	(than other visible guards), then that polygon has a \emph{conflict-free chromatic guarding} \cite{suri-conflict,Bartschi-2014,DBLP:journals/comgeo/HoffmannKSVW18}.
	If every guard that sees any given observer is of different types, then that polygon has a \emph{strong chromatic guarding} \cite{strong-conflict-free}. 
	
	\paragraph{Motivation.}
	
	In graph theory, a {\em conflict-free colouring} of a graph is an
	assignment of colours to vertices such that the
	closed neighbourhood of each vertex contains at least one unique colour.
	This problem was studied since 1973 by Biggs under the name \emph{perfect code}, which is
	a conflict-free colouring of a graph using only one colour, see \cite{biggs-1973,DBLP:conf/mfcs/KratochvilK88}.
	Later on, this topic arouse interest~on polygon visibility graphs when the field of robotics became widespread \cite{chazelle87a,gmr-sjc-97}.
	
	Consider a geometric scenario in which a mobile robot traverses a room from one point to another, communicating with the wireless sensors placed
	in the corners of the room.
	Even if the robot has full access to the map of the room, it cannot determine its location precisely because of accumulating rounding errors \cite{motionPlanningCourse}.
	And thus it needs clear markings in the room to guide itself to the end point in an energy efficient way.

	To guide a mobile robot with wireless sensors, two properties must be satisfied.
	First one is, no matter where the robot is in the polygon, it should hear from at least one sensor. That is, the placed sensors must together \emph{guard} the whole room and leave no place uncovered.
	The second one is, if the robot hears from several sensors, there must be at least one sensor broadcasting with a frequency that is not reused by some other sensor in the range. That is, the sensors must have \emph{conflict-free} frequencies.
	If these two properties are satisfied, then the robot can guide itself using the deployed wireless sensors as landmarks.
	This problem is also closely related to frequency assignment problem in wireless networks \cite{freqAssignment,suri-conflict}.
	One can easily solve this problem by placing a sensor at each corner of the room, and assigning a different frequency to each sensor.
	However, this method becomes very expensive as the number of sensors grow \cite{freqAssignment,cf-app}.
	Therefore, the main goal in this problem is minimise the number of different frequencies assigned to sensors.
	Since the cost of a sensor is comparatively very low, we do not aim to	minimise the number of sensors used.
	
	The above scenario is mathematically modelled as follows.
	The room is a simple polygon with $n$ vertices.
	There are $m$ sensors placed in the polygon (usually on some of its vertices),
	and two different sensors are given two different colours if, and only if they broadcast in different frequencies.

	\paragraph{Basic definitions.}

	A polygon $P$ is defined as a closed region $R$ in the plane bounded by a finite set
 	of line segments (called edges of $P$) \cite{g-vap-07}.
 	We consider simple polygons , i.e., simply connected regions
	(informally, ``without holes''), usually non-convex.
	Two points $p_1$ and $p_2$ of a polygon $P$ are said to \emph{see} each other, or be \emph{visible}
	to each other, if the line segment $\lseg{p_1}{p_2}$ belongs to~$P$.
\footnote{Note that our polygons are topologically closed, and hence
 	neighbouring vertices see each other along their common edge of the polygon.
	More generally, a straight line of sight can ``touch'' the boundary of the
	polygon, but not to cross it.
	In a physical-world scenario, that is when the
 	polygon is an empty area (formally topologically open) 
 	bounded by the surrounding walls, this definition
 	of visibility between $p_1$ and $p_2$ hence relates to an
 	arbitrarily small perturbation of the straight line of sight between $p_1$ and $p_2$.
}

	In this context, we say that a guard $g$ {\em guards} a point $x$ of $P$ if
	$g$ sees $x$, that is when the line segment $\lseg{g}{x}$ belongs to~$P$. 
	A vertex of $P$ is said to be a \emph{convex vertex} if its interior angle is less than $180^\circ$.
	Otherwise the vertex is said to be a \emph{reflex vertex}.
	A~polygon $P$ is a \emph{weak visibility polygon} if $P$ has an edge $uv$ 
	such that for every point $p$ of $P$ there is a point $p'$ on $uv$ seeing $p$. 

	A solution of {\em conflict-free chromatic guarding} of a polygon $P$
	consists of a set of {\em guards} in $P$,
	and an assignment of colours to the guards (one colour per guard) 
	such that the following holds;
	every {\em viewer} $v$ in $P$ (where $v$ can be any point of $P$ in
	our case, including a guard which sees itself) can see a guard of colour $c$ such
	that no other guard seen by $v$ has the same colour~$c$.
	In the {\em point-to-point} (P2P) variant the guards can be placed in any points of~$P$, while in
	the {\em vertex-to-point} (V2P) variant the guards can be placed only at the vertices of~$P$.
	There is also a {\em vertex-to-vertex} (V2V) variant in which both
	guards and viewers are restricted to the polygon vertices.
	We leave a discussion of the V2V variant till Section~\ref{sec:v2vcfc}.
	In all variants the goal is to minimise the number of colours (e.g., frequencies) used.

	When writing $\log n$, we mean the binary logarithm $\log_2n$.
	Further special definitions are presented in the coming sections in
	which they are used.

	\paragraph{Related research.}

	The aforementioned P2P conflict-free chromatic guarding (art gallery) problem has been studied in several papers. 
	B\"{a}rtschi and Suri gave an upper bound of $O(\log^2n)$ colours on simple
	$n$-vertex polygons \cite{suri-conflict}.
	Later, B\"{a}rtschi et al.\ improved this upper bound to $O(\log n)$
	on simple polygons \cite{Bartschi-2014}, and Hoffmann et al.~\cite{DBLP:journals/comgeo/HoffmannKSVW18},
	while studying the orthogonal variant of the problem, have given the first nontrivial lower bound of $\Omega\big(\log\log n/\log\log\log n\big)$ colours holding also in the general case of simple polygons.

	Our paper deals with the V2P variant
	in which guards should be placed on polygon vertices and viewers can be any points of the polygon.
	Note that there are some fundamental differences between point and
	vertex guards, e.g., funnel polygons (Fig.~\ref{fig:funnel}) can
	always be guarded by one point guard (of one colour) but they may require up to 
	$\Omega (\log n)$ colours in the V2P conflict-free chromatic
	guarding, as shown in~\cite{Bartschi-2014}.
	Hence, extending a general upper bound of $O(\log n)$ colours for point guards on simple polygons by B\"{a}rtschi et al.~\cite{Bartschi-2014} to the more restrictive vertex guards is a challenge, 
	which we can now only approach with an $O(\log^2n)$ bound (see below). 
	Note that the same bound of $O(\log^2n)$ colours was attained by 
	\cite{Bartschi-2014} when allowing multiple guards at the same vertex.
	
	\paragraph*{Our results}
	\begin{enumerate}
	\item (Sections~\ref{sec:vpguarding} and~\ref{sec:funnelcolouring}, in Theorems \ref{thm:sizeguard-funnel}, \ref{thm:logm-4} and~\ref{cor:opt+4})~
	We give a polynomial-time algorithm to find the optimum number $m$ of vertex-guards to guard all the points of a funnel, and show that the number of 	
	colours in the corresponding conflict-free chromatic guarding problem is $\log m+O(1)$.
	This leads to an approximation algorithm for V2P conflict-free chromatic guarding of a funnel, with only a constant ($+4$) additive error.
	A remarkable feature of this result is that we prove a direct relation (Theorem~\ref{thm:logm-4}) between the optimal numbers of guards and of colours 	needed in funnel polygons. \label{it:result1}
	\item (Sections~\ref{sec:v2pChromatic} and~\ref{sec:allsimple} in Theorems \ref{thm:twoApprox} and \ref{thm:allpolygon})
	We show that a weak visibility polygon on $n$ vertices can be V2P conflict-free chromatic guarded with only $O (\log^2 n)$
	guards, by decomposing the problem into funnels and colouring each funnel by $O (\log n)$ colours.
	We generalise this upper bound to all simple polygons, which is a result incomparable with previous~\cite{Bartschi-2014}. \label{it:result2}
	\item (Section~\ref{sec:v2vcfc} in Proposition~\ref{pro:v2v-3colours} and Theorem \ref{thm:v2vhardness})
	We show that a V2V chromatic guarding may require at least three colours,
	and conjecture that three colours always suffice in a simple polygon
	(unlike in the V2P and P2P settings for which this is unbounded~\cite{Bartschi-2014,DBLP:journals/comgeo/HoffmannKSVW18}).
	We prove that determining whether a given polygon has a V2V conflict-free
	chromatic guard set using only one or two colours is NP-complete.
	
	\end{enumerate}
Note that all our algorithms are simple and suitable for easy implementation.

	\section{Minimising Vertex-to-Point Guards for Funnels} \label{sec:vpguarding}
	In the next two sections, we focus on a special interesting type of polygons -- {\em funnels}.
	A polygon $P$ is a funnel if precisely three of the vertices of $P$ are convex, and two of the convex vertices share one common edge -- 
	the {\em base} of the funnel~$P$.
	
	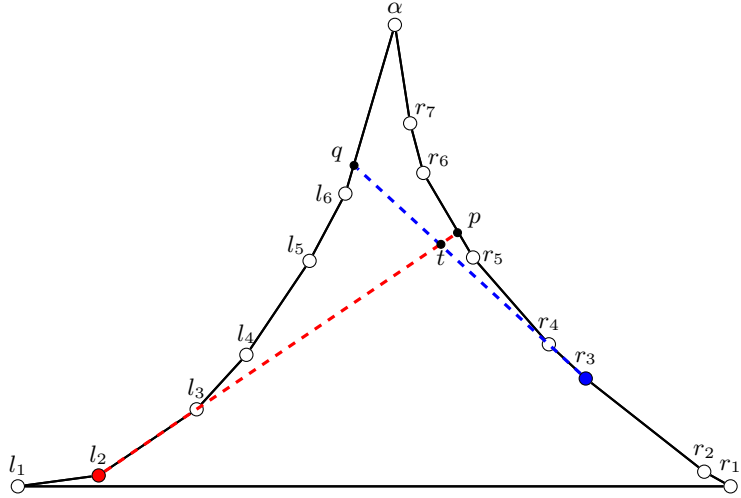
\begin{figure}[tbp]
		\centering
		\begin{tikzpicture}[xscale=0.35, yscale=0.5]
		
		\tikzstyle{every node}=[draw, shape=circle, minimum size=5pt,inner sep=0pt];
		\node[label=$l_1$] (A) at (1.93,4.71) {};
		\node[label=$l_2$,fill=red] (B) at (5,5) {};
		\node[label=$l_3$] (C) at (8.72,6.76) {};
		\node[label=$l_4$] (D) at (10.61,8.21) {};
		\node[label=120:$l_5$] (E) at (13.01,10.71) {};
		\node[label=left:$l_6$] (F) at (14.37,12.5) {};
		\node[label=$\alpha$] (G) at (16.25,16.99) {};
		\node[label=30:$r_7$] (H) at (16.83,14.37) {};
		\node[label=30:$r_6$] (I) at (17.33,13.05) {};
		\node[label=right:$r_5$] (J) at (19.22,10.8) {};
		\node[label=$r_4$] (K) at (22.1,8.49) {};
		\node[label=$r_3$,fill=blue] (L) at (23.5,7.58) {};
		\node[label=$r_2$] (M) at (28,5.1) {};
		\node[label=$r_1$] (N) at (29,4.71) {};
		
		\draw[thick](A)--(B)--(C)--(D)--(E)--(F)--(G)--(H)--(I)--(J)--(K)--(L)--(M)--(N)--(A);
		
		\coordinate (A') at (24.71,6.89);
		\node[minimum size = 3pt, fill=black, label=40:$~p$] (B') at (18.63, 11.46) {};
		\coordinate (C') at (17.25, 13.27);
		\node[minimum size = 3pt, fill=black, label=150:$q~$] (M') at (14.7,13.25) {};
		
		\draw[thick](A)--(B)--(C)--(D)--(E)--(F)--(G)--(H)--(I)--(J)--(K)--(L)--(M)--(N)--(A);
		\draw[very thick, color=red, dashed] (B)--(B');
		\draw[very thick, color=blue, dashed] (L)--(M');
		\node[minimum size = 3pt, fill=black, label=270:$t$] (B'M') at (18,11.15) {};
		\end{tikzpicture}
		
		\caption{A funnel $F$ with seven vertices in $\mathcal{L}$ labelled $l_1,
			\ldots, l_7$ from bottom to top, and eight vertices in $\mathcal{R}$ 
			labelled $r_1, \ldots, r_8$, including the apex $\alpha = l_7 = r_8$.
			All vertices of $F$ except $l_1,r_1$ and $\alpha$ are reflex vertices.
			The picture also shows the upper tangent of the vertex $l_2$ of $\mathcal L$ (drawn in
			dashed red), the upper tangent of the vertex $r_3$ of $\mathcal R$
			(drawn in dashed blue), and their intersection~$t$.}
		\label{fig:funnel}
		\label{fig:tangents}
	\end{figure}

	\label{su:numguardsfun}
	Before turning to the conflict-free chromatic guarding problem, we
	first resolve the problem of minimising the total number of
	vertex guards needed to guard all points of a funnel.	
	We start by describing a simple procedure (Algorithm~\ref{alg:tightpath})
	that provides us with a guard set which may not always be optimal (but
	very close to the optimum, see Lemma~\ref{lem:alg1-near-optimal}). 
	This procedure will be helpful for the subsequent colouring results.
	Then, we also refine the simple procedure to compute the optimal number of guards 
	in Algorithm~\ref{alg:newminfunguard}.
	
	We use some special notation here. See Figure~\ref{fig:funnel}.
	Let the given {funnel} be~$F$, oriented in the plane as follows. 
	On the bottom, there is the horizontal \emph{base} of the funnel -- the line segment $\lseg{l_1}{r_1}$ in the picture.
	The topmost vertex of $F$ is called the \emph{apex}, and it is denoted by $\alpha$.
	There always exists a point $x$ on the base which can see the apex $\alpha$, and then $x$ sees the whole funnel at once.
	The vertices on the left side of apex form the \emph{left concave chain}, and analogously, the vertices on the right side of the apex form the 				\emph{right concave chain} of the funnel.
	These left and right concave chains are denoted by $\mathcal{L}$ and $\mathcal{R}$ respectively.
	We denote the vertices of $\mathcal{L}$ as $l_1, l_2, \ldots, l_k$ from bottom to top.
	We denote the vertices of $\mathcal{R}$ as $r_1, r_2, \ldots, r_m$ from bottom to top.
	Hence, the apex is $l_k = r_m = \alpha$.

	Let $l_i$ be a vertex on $\mathcal{L}$ which is not the apex. 
	We define the \emph{upper tangent} of $l_i$, denoted by $\uptan(l_i)$, as the ray whose origin is $l_i$ and which passes through $l_{i+1}$.
	Upper tangents for vertices on $\mathcal{R}$ are defined analogously. 
	Let $p$ be the point of intersection of $\mathcal{R}$ and the upper tangent of~$l_i$.
	Then we define $\upseg(l_i)$ as the line segment $\lseg{l_{i+1}}p$. 
	For the vertices of $\mathcal{R}$, $\upseg$ is defined analogously: 
	if $q$ is the point of intersection of $\mathcal{L}$ and the upper tangent of
	$r_j \in \mathcal{R}$, then let $\upseg(r_j) := \lseg{r_{j+1}}q$.
	See again Figure~\ref{fig:tangents}.
	
	\begin{algorithm}[tbp]
		\KwIn{A funnel $F$ with concave chains $\mathcal{L}=(l_1, 
			\ldots, l_k)$ and $\mathcal{R}=(r_1, \ldots, r_m)$.}
		\KwOut{A vertex set guarding all the points of $F$.}
		\smallskip
		
		Initialise an auxiliary digraph $G$ with two dummy vertices $x$ and $y$,
		\mbox{and declare $\upseg(x) = \lseg{l_1}{r_1}$}\;
		Initialise $\ES{S} \gets \{x\}$\;
		\While {$\ES{S}$ is not empty}
		{
			Choose an arbitrary $t \in \ES{S}$, and remove $t$ from $\ES{S}$\;
			Let $s=\upseg(t)$	\tcc*{$s$ is a segment inside $F$}
			Let $q$ and $p$ be the ends of $s$ on $\mathcal{L}$ and $\mathcal{R}$, respectively\;
			Let $i$ and $j$ be the largest indices such that $l_i$ and
			$r_j$ are not above $q$ and~$p$, resp.\;
			\lIf{$l_{i+1}$ can see whole $s$}{$i' \gets i+1$}
			\lElse(\tcc*[f]{the topmost vertex on the left seeing whole $s$})
			{$i' \gets i$}
			\lIf{$r_{j+1}$ can see whole $s$}{$j' \gets j+1$}
			\lElse(\tcc*[f]{the topmost vertex on the right seeing whole $s$})
			{$j' \gets j$}
			Include the vertices $l_{i'}$ and $r_{j'}$ in~$G$\;
			\ForEach{$z\in\{l_{i'},r_{j'}\}$}
			{\label{lin:zin}%
				Add the directed edge $(t,z)$ to~$G$ \label{lin:wght}\;
				\If{segment $\upseg(z)$ includes the apex $l_k=r_m$}
				{Add the directed edge $(z,y)$ to~$G$
					\tcc*{$y$ is the dummy vertex}}
				\lElse(\tcc*[f]{more guards are needed above~$z$})
				{$\ES{S}\gets \ES{S}\cup\{z\}$}
			}
		}
		Enumerate a shortest path from $x$ to $y$ in~$G$\; 
		Output the shortest path vertices without $x$ and $y$ as the required guard set\;
		\caption{Simple vertex-to-point guarding of funnels (uncol.)}
		\label{alg:tightpath}
	\end{algorithm}
	
	The underlying idea of Algorithm~\ref{alg:tightpath} is as follows.
	Imagine we proceed bottom-up when building the guard set of a funnel $F$.
	Then the next guard is placed at the top-most vertex $z$ of $F$,
	nondeterministically choosing between $z$ on the left and the right chain of
	$F$, such that no ``unguarded gap'' remains below~$z$.
	Note that the unguarded region of $F$ after placing a guard at $z$
	is bounded from below by~$\upseg(z)$.
	The nondeterministic choice of the next guard~$z$ is encoded within a digraph,
	in which we then find the desired guard set as a shortest path.
	The following is straightforward.
	
	\begin{lemma}\label{lem:tightpath-feas}
		Algorithm~\ref{alg:tightpath} runs in polynomial time, and it outputs
		a feasible guard set for all the points of a funnel $F$.
	\end{lemma}
	\begin{proof}
	As for the runtime, we observe that the number of considered line segments
	$s$ in the algorithm is, by the definition of $\upseg$, 
	bounded by at most $k+m$ (and it is typically much lower than this bound).
	Each considered segment $\upseg(t)$ of $t \in \ES{S}$ is processed at most once, and it contributes two edges to $G$.
	Overall, a shortest path in $G$ is found in linear time.
	
	We prove feasibility of the output set by induction.
	Let $(x=x_0,x_1,\ldots,x_{a-1},$ $x_a=y)$ be a path in~$G$.
	We claim that, for $0\leq i\leq a$, guards placed at $x_0,x_1,\ldots,x_i$ guard all the points of $F$ below $\upseg(x_i)$. 
	This is trivial for $i=0$, and it straightforwardly follows by induction:
	Algorithm~\ref{alg:tightpath} asserts that $x_i$ can see whole $\upseg(x_{i-1})$, and $x_i$ hence also sees the strip between $\upseg(x_{i-1})$ 			and $\upseg(x_{i})$ by basic properties of a funnel.
	Finally, at $x_{a-1}$, we guard whole $F$ up to its apex.
	\end{proof}
	
	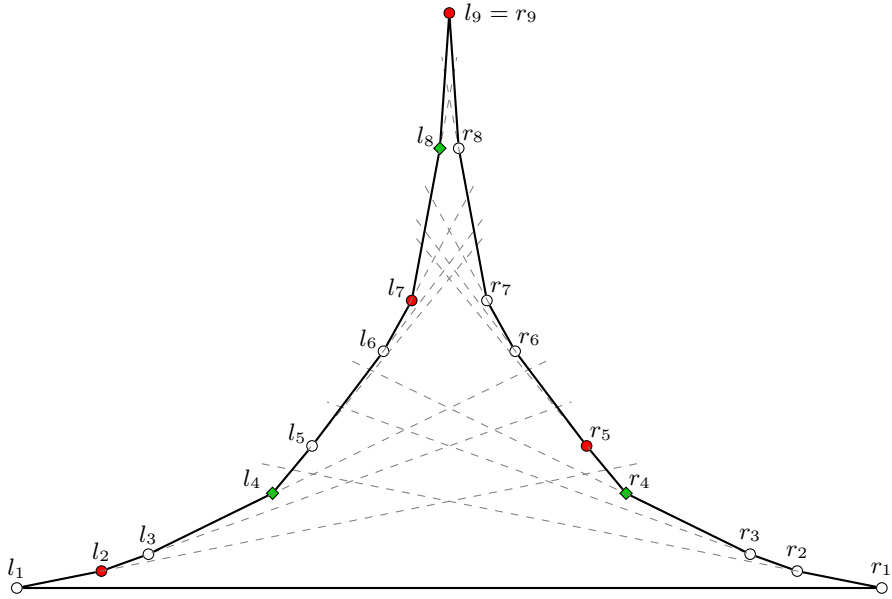
\begin{figure}[btp]
		\centering
		\begin{tikzpicture}[xscale=0.25, yscale= 0.45]
		
		\tikzstyle{every node}=[draw, shape=circle, minimum size=4pt,inner sep=0pt];
		\node[label=$l_1$] (A) at (-23,0) {};
		\node[fill=red, label=$l_2$] (B) at (-18.5,0.5) {};
		\node[label=$l_3$] (C) at (-16,1) {};
		\node[fill=green!80!black, label=150:$l_4~$, diamond,inner sep=1.2pt] (D) at (-9.4,2.8) {};
		\node[label=150:$l_5$] (E) at (-7.3,4.2) {};
		\node[label=150:$l_6$] (F) at (-3.5,7) {};
		\node[fill=red, label=150:$l_7$] (G) at (-2,8.5) {};
		\node[fill=green!80!black, label=150:$l_8$, diamond,inner sep=1.2pt] (H) at (-0.5,13) {};
		\node[fill=red, label=right:{$~l_9=r_9$}] (I) at (0,17) {};
		
		\node[label=$r_1$] (Q) at (23,0) {};
		\node[label=$r_2$] (P) at (18.5,0.5) {};
		\node[label=$r_3$] (O) at (16,1) {};
		\node[fill=green!80!black, label=30:$r_4$, diamond,inner sep=1.2pt] (N) at (9.4,2.8) {};
		\node[fill=red, label=30:$r_5$] (M) at (7.3,4.2) {};
		\node[label=30:$r_6$] (L) at (3.5,7) {};
		\node[label=30:$r_7$] (K) at (2,8.5) {};
		\node[label=30:$r_8$] (J) at (0.5,13) {};
		
		\coordinate (A') at (10.2,3.7);
		\coordinate (Q') at (-10.2,3.7);
		
		\coordinate (B') at (6.5,5.5);
		\coordinate (P') at (-6.5,5.5);
		\coordinate (C') at (5.5,6.8);
		\coordinate (O') at (-5.5,6.8);
		
		\coordinate (D') at (2,10.5);
		\coordinate (N') at (-2,10.5);
		
		\coordinate (E') at (1.9, 11);
		\coordinate (M') at (-1.9, 11);
		\coordinate (F') at (1.4, 12);
		\coordinate (L') at (-1.4, 12);
		
		\coordinate (G') at (0.4,15.7);
		\coordinate (K') at (-0.4,15.7);
		
		\coordinate (H') at (3.61, 5.6);
		\coordinate (J') at (-3.61, 5.6);
		
		\draw[color=gray, dashed] (A)--(A');
		\draw[color=gray, dashed] (Q)--(Q');
		\draw[color=gray, dashed] (D)--(D');
		\draw[color=gray, dashed] (N)--(N');
		\draw[color=gray, dashed] (E)--(E');
		\draw[color=gray, dashed] (F)--(F');
		\draw[color=gray, dashed] (M)--(M');
		\draw[color=gray, dashed] (L)--(L');
		\draw[color=gray, dashed] (B)--(B');
		\draw[color=gray, dashed] (C)--(C');
		\draw[color=gray, dashed] (P)--(P');
		\draw[color=gray, dashed] (O)--(O');
		\draw[color=gray, dashed] (G)--(G');
		\draw[color=gray, dashed] (K)--(K');
		
		\draw[thick]  (A)--(B)--(C)--(D)--(E)--(F)--(G)--(H)--(I)--(J)--(K)--(L)--(M)--(N)--(O)--(P)--(Q)--(A);
	
		\end{tikzpicture}
		\caption{A symmetric funnel with $17$ vertices. 
			The gray dashed lines show the upper tangents of the vertices.
			It is easy to see that Algorithm~\ref{alg:tightpath} selects $4$ guards,
			up to symmetry, at $l_2,r_5,l_7,l_9$ (the red vertices).
			However, the whole funnel can be guarded by three guards
			at $l_4,r_4,l_8$ (the green diamond vertices), and it will be the task of
			Algorithm~\ref{alg:newminfunguard} to consider such better possibility.}
		\label{fig:funnelguards}
	\end{figure}
	
	\begin{remark}
	Unfortunately, the guard set produced by Algorithm~\ref{alg:tightpath} may not be optimal under certain circumstances.
	See the example in Figure~\ref{fig:funnelguards}; the algorithm picks the four red vertices, but the funnel can be guarded by the three green 				vertices. 
	Nevertheless, this (possibly non-optimal) guard set will be very useful in the next section in the context of conflict-free colouring.
	\end{remark}
		
	For the sake of completeness, we now refine the simple approach of Algorithm~\ref{alg:tightpath} to always produce a minimum size guard set.
	Recall that Algorithm~\ref{alg:tightpath} always places one next guard based on the position of the previous one guard.
	Our refinement is going to consider also pairs of guards (one from the left and one from the right chain) in the procedure. 
	We correspondingly extend the definition of $\upseg$ to pairs of vertices as follows.
	Let $l_i$ and $r_j$ be vertices of $F$ on $\mathcal{L}$ and $\mathcal{R}$, respectively, such that $\upseg(l_i)=\lseg{l_{i+1}}p$ intersects 
	$\upseg(r_j)=\lseg{r_{j+1}}q$ in a point~$t$ (see in Figure~\ref{fig:tangents}).
	Then we set $\upseg(l_i,r_j)$ as the polygonal line (``$\vee$-shape'') $\lseg pt\cup\lseg qt$.
	In case that $\upseg(l_i)\cap\upseg(r_j)=\emptyset$, we simply define
	$\upseg(l_i,r_j)$ as the upper one of $\upseg(l_i)$ and	$\upseg(r_j)$.

	\begin{algorithm}[tbp]
		\KwIn{A funnel $F$ with concave chains $\mathcal{L}=(l_1, l_2,
			\ldots, l_k)$ and $\mathcal{R}=(r_1, \ldots, r_m)$.}
		\KwOut{A minimum vertex set guarding all the points of $F$.}
		\smallskip
		Initialise an auxiliary digraph $G$ with two dummy vertices $x$ and $y$,
		\mbox{and declare $\upseg(x) = \lseg{l_1}{r_1}$}\;
		Initialise $\ES{S} \gets \{x\}$\;
		\While {$\ES{S}$ is not empty}
		{
			Choose an arbitrary $t \in \ES{S}$, and remove $t$ from $\ES{S}$\;
			Let $s=\upseg(t)$
			\tcc*{$s$ is a segment or a $\vee$-shape}
			Let $i'$ and $j'$ be defined for $s$ as in Algorithm~\ref{alg:tightpath}\;
			
			Let $q$ and $p$ be the ends of $s$ on $\mathcal{L}$ and $\mathcal{R}$, respectively\;
			Let $i''$ and $j''$ be the largest indices such that
			\mbox{$l_{i''}$ lies strictly below $\upseg(p)$ and 
				$r_{j''}$ strictly below $\upseg(q)$}\;
			\label{line:7}
			\tcc*[f]{then $l_{i''}$ and $r_{j''}$ together can see whole $s$}
			
			Include the vertices $l_{i'}$, $r_{j'}$ and $(l_{i''},r_{j''})$ in~$G$\;
			
			\ForEach{$z\in\{l_{i'},r_{j'},(l_{i''},r_{j''})\}$}
			{
				Add the directed edge $(t,z)$ to~$G$, and
				\mbox{\quad assign $(t,z)$ weight $2$ if $z=(l_{i''},r_{j''})$,
					and weight $1$ otherwise}\;
			\label{line:9}
				\If{$\upseg(z)$ includes the apex $l_k=r_m$}
				{Add the directed edge $(z,y)$ to~$G$ of weight~$0$\;}
				\lElse(\tcc*[f]{more guards are needed above~$z$})
				{$\ES{S}\gets \ES{S}\cup\{z\}$}
			}
		}
		Enumerate a shortest weighted path from $x$ to $y$ in~$G$\; 
		Output the shortest path vertices without $x$ and $y$,
		but considering the possible guard pairs, as the required guard set\;
		
		\caption{Optimum vertex-to-point guarding of funnels (uncol.)}
		\label{alg:newminfunguard}
	\end{algorithm}

	Algorithm~\ref{alg:newminfunguard}, informally saying,
	enriches the two nondeterministic choices of placing the next guard
	in Algorithm~\ref{alg:tightpath} with a third choice;
	placing a suitable top-most pair of guards $z=(z_1,z_2)$, 
	$z_1\in\mathcal{L}$ and $z_2\in\mathcal{R}$,
	such that again no ``unguarded gap'' remains below $(z_1,z_2)$.
	Figure~\ref{fig:funnelguards} features a funnel in which placing such a pair
	of guards $(z_1=l_4,\>z_2=r_4)$ may be strictly better than using any two
	consecutive steps of Algorithm~\ref{alg:tightpath}.
	On the other hand, we can show that there is no better possibility than one of these three considered steps.
	Within the scope of Algorithm~\ref{alg:newminfunguard} (cf.~line~\ref{line:7}),
	we extend the definition range of $\upseg(\cdot)$
	to include all boundary points of $\mathcal{L}$ and $\mathcal{R}$, as follows.
	If $p$ is an internal point of $\lseg{l_i}{l_{i+1}}$, then we set
	$\upseg(p):=\upseg(l_i)$.
	If $p'$ is an internal point of $\lseg{r_j}{r_{j+1}}$, then we set
	$\upseg(p'):=\upseg(r_j)$.
	
	\begin{theorem}\label{thm:sizeguard-funnel}
		Algorithm~\ref{alg:newminfunguard} runs in polynomial time, and it outputs
		a feasible guard set of minimum size guarding all the points of a funnel $F$.
	\end{theorem}
	\begin{proof}
		Proving polynomial runtime is analogous to Lemma~\ref{lem:tightpath-feas},
		only now obtaining a quadratic worst-scenario bound.
		Likewise the proof of feasibility of the obtained solution is analogous to
		the previous proof. 
		We only need to observe the following new claim:
		if $i''$and $j''$ are defined as on line~\ref{line:7} of
		Algorithm~\ref{alg:newminfunguard}, then $l_{i''}$ and $r_{j''}$ together
		can see whole $s$ and the strip of $F$ from $s$ till
		$\upseg(l_{i''},r_{j''})$.
		The crucial part is to prove optimality.
		
		Having two guard sets $A,B\subseteq V(F)$, we say that $A$ {\em covers} $B$ if there is an injection $c:A\to B$ such that, for each $a\in A$,
		the guard $c(a)$ is on the same (left or right) chain of~$F$ as $a$ and not higher than~$a$.
		Let $G$ be the digraph constructed by Algorithm~\ref{alg:newminfunguard} on~$F$. 
		Note that the weight of any $x$--$y$ path in $G$ equals the number of guards placed along it.
		Hence, together with claimed feasibility of the computed solution, it is enough to prove:
		\begin{itemize}
			\item For every feasible vertex guard set $D$ of $F$,
			there exists a feasible guard set $A$ which covers $D$, and
			$A$ has a corresponding directed $x$--$y$ path in the graph~$G$
			of Algorithm~\ref{alg:newminfunguard}.
		\end{itemize}
		
		Let $A$ be any feasible guard set of $F$ which covers given $D$,
		and such that $A$ is maximal w.r.t.~the cover relation.
		Let $P$ be a maximal directed path in $G$, starting from $x$, such that the set of guards $B_P$ listed in the vertices of $P$ (without~$x$)
		satisfies~$B_P\subseteq A$.
		Obviously, we aim to show that $P$ ends in~$y$.
		Suppose not (it may even be that $P$ is a single vertex~$x$ and $B_P=\emptyset$).
		Let $t$ be the last vertex of $P$ and denote by $s=\upseg(t)$ and let $q$ and $p$ be the ends of $s$ on $\mathcal{L}$ and $\mathcal{R}$,
		respectively, as in the algorithm.
		
		Let $A'=A\setminus B_P\not=\emptyset$.
		Then $s$ has to be guarded from $A'$ (while the whole part of $F$ below $s$ is already guarded by $B_P$ by feasibility of the algorithm).
		Let $i,j$ be such that $l_i\in A'\cap\mathcal{L}$ and $r_j\in A'\cap\mathcal{R}$ are the lowest guards on the left and right chain.
		Assume, up to symmetry, that $l_i$ can see whole~$s$.
		By our maximal choice of $A$ we have that no vertex on $\mathcal{L}$ above $l_i$ can see whole~$s$,
		and so the digraph $G$ contains an edge from $t$ to $l_i$ (line~\ref{line:9} of Algorithm~\ref{alg:newminfunguard}),
		which contradicts maximality of the path~$P$.
		
		Otherwise, neither of $l_i$, $r_j$ can see whole $s$,
		and so $l_i$ sees the end $p$ and $r_j$ sees the end~$q$.
		Consequently, $l_{i}$ is strictly below $\upseg(p)$ and            
		$r_{j}$ strictly below $\upseg(q)$, and they are topmost such vertices
		again by our maximal choice of $A$.
		Hence the digraph $G$ contains an edge from $t$ to $(l_i,r_j)$,
		as previously, which is again a contradiction concluding the proof.
	\end{proof}
	
	Lastly, we establish that the difference between Algorithms \ref{alg:tightpath} and \ref{alg:newminfunguard} 
	cannot be larger than~$1$ guard,
	because we would like to use the simpler Algorithm~\ref{alg:tightpath} instead of the latter one in subsequent applications.

	\begin{lemma}\label{lem:alg1-near-optimal}
		The guard set produced by Algorithm \ref{alg:tightpath} is always by at most one
		guard larger than the optimum solution produced by Algorithm~\ref{alg:newminfunguard}.
	\end{lemma}
	
	\begin{proof}
	Let $G^1$ with the source $x^1$ be the auxiliary graph produced by Algorithm~\ref{alg:tightpath},
	and $G^2$ with the source $x^2$ be the one produced by Algorithm~\ref{alg:newminfunguard}.
	We instead prove the following refined statement by induction on~$i\geq0$.
	Recall the detailed inductive statement we are going to prove now:
		\begin{itemize}
			\item Let $P^2=(x^2=x^2_0,x^2_1,\ldots,x^2_i)$ be any directed path in $G^2$ 
			of weight $k$, let $Q^2$ denote the set of guards listed in the vertices of $P^2$, and
			$L^2=\mathcal{L}\cap Q^2$ and $R^2=\mathcal{R}\cap Q^2$.
			Then there exists a directed path $(x^1=x^1_0,x^1_1,\ldots,x^1_k,x^1_{k+1})$ in $G^1$
			(of length $k+1$), such that the guard of $x_k$ is at least as high as all the
			guards of $L^2$ (if $x_k\in\mathcal{L}$) or of $R^2$ (if
			$x_k\in\mathcal{R}$), and the guard of $x_{k+1}$ is strictly higher than all
			the guards of $Q^2$.
		\end{itemize}
		
		The claim is trivial for $i=0$, and so we assume that $i\geq1$ and the claim
		holds for the shorter path ${P^2}'=(x^2=x^2_0,x^2_1,\ldots,x^2_{i-1})$ of
		weight $k'$ in $G^2$, hence providing us with a path
		$(x^1=x^1_0,x^1_1,\ldots,x^1_{k'},x^1_{k'+1})$ in $G^1$.
		If $k=k'+1$ (i.e., $x^2_i$ represents a single guard),
		Algorithm~\ref{alg:tightpath} can ``duplicate'' the move, hence making
		$x^2_i$ or a higher vertex $x^2_{i'}$ on the same chain 
		an out-neighbour of $x^1_{k'+1}$ in $G^1$.
		Then we set $x^1_{k'+2}=x^1_{k+1}=x^2_{i'}$ and we are done.
		
		If $k=k'+2$ (i.e., $x^2_i$ represents a pair of guards $z_1,z_2$),
		we proceed as follows.
		Up to symmetry, assume $x^1_{k'+1}\in\mathcal{L}$ 
		and $z_1\in\mathcal{L}$, $z_2\in\mathcal{R}$.
		By the induction assumption, we know that $x^1_{k'+1}$ is strictly higher
		(on $\mathcal{L}$) than the guards from $L^2\setminus\{z_1\}$.
		We choose $x^1_k$ as the out-neighbour of $x^1_{k'+1}$ in $G^1$ that lies on
		$\mathcal{R}$, and $x^1_{k+1}$ as the out-neighbour of $x^1_{k}$ in
		$G^1$ that lies back on $\mathcal{L}$. 
		From Algorithm~\ref{alg:newminfunguard} (line~\ref{line:7})
		it follows that $z_2$ sees $x^1_{k'+1}$, and so $x^1_k$ is at least as high
		on $\mathcal{R}$ as~$z_2$.
		Consequently, $x^1_{k+1}$ lies on $\mathcal{L}$ strictly higher than~$z_1$
		(which sees the highest guard from $R^2\setminus\{z_2\}$), and we are again done.
	\end{proof}

	\section{Vertex-to-Point Conflict-Free Chromatic Guarding of Funnels}
	 \label{sec:funnelcolouring}
	
	In this section, we continue to study funnels.
	To obtain a conflict-free coloured solution, we will simply
	consider the guards chosen by Algorithm~\ref{alg:tightpath} in the ascending
	order of their vertical coordinates, and colour them in the \emph{ruler sequence},
	(e.g.,~\cite{guy-1994}) in which the $i^{th}$ term is the exponent of the largest power of $2$ that divides $2i$. 
	(The first few terms of it are $1,2,1,3,1,2,1,4,1,2,1,3,1,2,$ $1,5,1,2,1,3\dots$.)
	So, if Algorithm~\ref{alg:tightpath} gives $m$ guards, then our approach
	will use about $\log m$ colours.
	
	Our aim is to show that this is always very close to the optimum, by giving a lower bound
	on the number of necessary colours of order $\log m-O(1)$.
	To achieve this, we study the following two sets of guards for a given funnel~$F$:
	\begin{itemize}
		\item The minimal {\em guard set $A$} computed by Algorithm~\ref{alg:tightpath} on~$F$
		(which is overall nearly optimal by Lemma~\ref{lem:alg1-near-optimal});
		if this is not unique, then we fix any such~$A$.
		\item A {\em guard set $D$} which achieves the minimum number of colours for
		conflict-free guarding; note that $D$ may be much larger than $A$ since it
		is the number of colours which matters.
	\end{itemize}
	On a high level, we are going to show that the colouring of $D$
	must (somehow) copy the ruler sequence on~$A$.
	For that we will recursively bisect our funnel into smaller ``layers'',
	gaining one unique colour with each bisection.
	
	Recall that our funnel $F$ consists of the concave chains $\mathcal{L}=(l_1, l_2,
	\ldots, l_k)$ and $\mathcal{R}=(r_1, r_2, \ldots, r_m)$.
	Analogously to the notion of an upper tangent from
	Section~\ref{su:numguardsfun}, we define the {\em lower tangent}
	of a vertex $l_i\in\mathcal{L}$, denote by $\lowtan(l_i)$,
	as the ray whose origin is $l_i$ and which passes through~$r_j\in\mathcal{R}$
	such that $r_j$ is the lowest vertex on $\mathcal{R}$ seeing $l_i$.
	Note that $\lowtan(l_i)$ may intersect $\mathcal{R}$ in $r_j$ alone or in a
	segment $\overline{r_jr_{j+1}}$.
	Let $\lowseg(l_i):=\lseg{l_i}{r_j}$.
	The definition of $\lowtan()$ and $\lowseg()$ for vertices of $\mathcal{R}$ is
	symmetric.
	
	We now give a definition of ``layers'' of a funnel which is the
	first crucial term for our proof.
	
	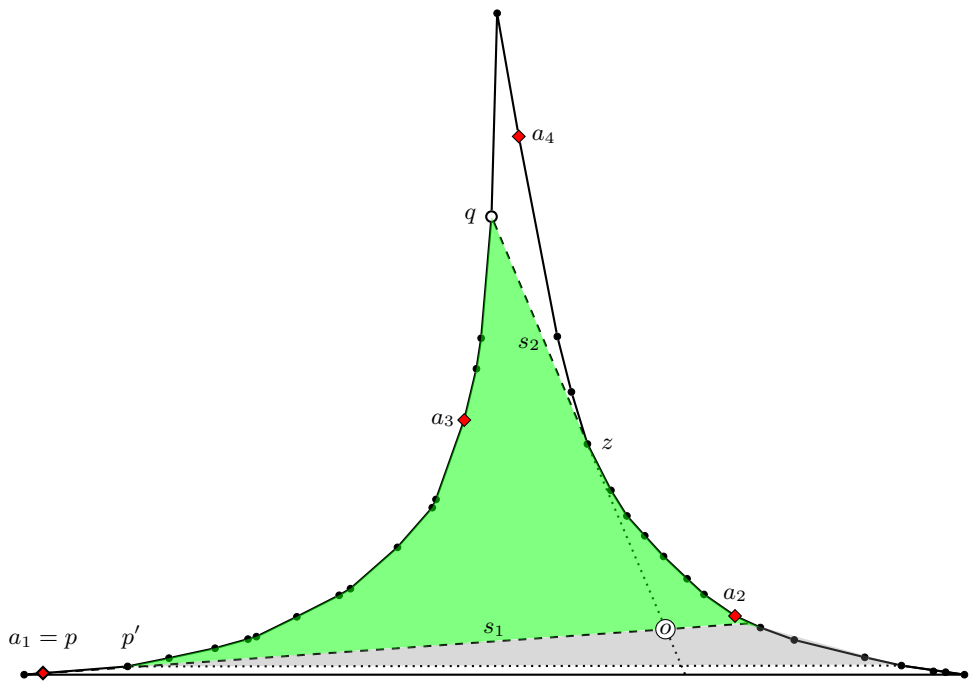
\begin{figure}[tbp]
		\centering
		\begin{tikzpicture}[yscale=1.1, xscale=1.25]
		
		\tikzstyle{every node}=[draw, shape=circle, fill=black, minimum size=2.5pt,inner sep=0pt];
		\node (1) at (-5,0) {};
		\node[label=above:~~$p'$~~] (2) at (-3.9, 0.1) {};
		\node (3) at (-3.46, 0.2) {};
		\node (4) at (-2.97, 0.32) {};
		\node (5) at (-2.62, 0.43) {};
		\node (6) at (-2.53, 0.46) {};
		\node (7) at (-2.1, 0.7) {};
		\node (8) at (-1.65, 0.96) {};
		\node (9) at (-1.53, 1.04) {};
		\node (10) at (-1.03, 1.54) {};
		\node (11) at (-0.66, 2.02) {};
		\node (12) at (-0.62, 2.12) {};
		\node (13) at (-0.32, 3.08) {};
		\node (14) at (-0.19, 3.7) {};
		\node (15) at (-0.14, 4.07) {};
		\node[label=left:$q\>$, fill=white,thick,minimum size=4pt] (16) at (-0.03, 5.54) {};
		\node (17) at (0.03, 8) {};
		\node[fill=red,minimum size=5pt,diamond, label=right:$\,a_4$] (18) at (0.26, 6.51) {};
		\node (19) at (0.67, 4.09) {};
		\node (20) at (0.82, 3.42) {};
		\node[label=right:~$z$] (21) at (0.99, 2.79) {};
		\node (22) at (1.24, 2.23) {};
		\node (23) at (1.41, 1.92) {};
		\node (24) at (1.6, 1.68) {};
		\node (25) at (1.8, 1.43) {};
		\node (26) at (2.05, 1.16) {};
		\node (27) at (2.23, 0.97) {};
		\node (28) at (2.56, 0.71) {};
		\node (29) at (2.83, 0.57) {};
		\node (30) at (3.19, 0.42) {};
		\node (31) at (3.94, 0.21) {};
		\node (32) at (4.33, 0.11) {};
		\node (33a) at (4.67, 0.04) {};
		\node (33b) at (4.8, 0.03) {};
		\node (33) at (5, 0) {};
		\foreach \i in {1,...,32} 
		{
			\pgfmathtruncatemacro\j{\i+1};
			\draw[thick] (\i)--(\j);
		}
		\draw[thick] (32)--(33)--(1);
		
		\coordinate(1') at (2.7,0.62);
		\coordinate(16') at (2.03, 0);
		\coordinate(1int) at (2.6,0.62); %intersection points]
		\coordinate(16int) at (2.02,0);
		\draw[dashed, thick] (1)--(1');
		\draw[dashed, thick] (16)--(21);
		\draw[dotted, thick] (21)--(16');
		\draw[dotted, thick] (2)--(32);
		\fill[opacity=0.3,gray] (2) -- (32) -- (1') -- (2);
		\fill[opacity=0.5,green] (1'.center)--(2.center)--(3.center)--(4.center)--(5.center)--(6.center)--(7.center)--(8.center)--(9.center)--(10.center)--(11.center)--(12.center)--(13.center)--(14.center)--(15.center)--(16.center)--(21.center)--(22.center)--(23.center)--(24.center)--(25.center)--(26.center)--(27.center)--(28.center)--(29.center)--(1int)--(1.center);

		\node[fill=white, inner sep=0.5pt] at (1.82,0.55) {{\small $o$}};
		\node[fill=red,minimum size=5pt,thick,diamond, label=above:{$\!\!\!a_1=p\!\!\!$}] (1g) at (-4.8,0.02) {};
		\node[fill=red,minimum size=5pt,diamond, label=left:$a_3$] (13g) at (-0.32, 3.08) {};
		\node[fill=red,minimum size=5pt,diamond, label=above:$a_2$] (28g) at (2.56, 0.71) {};
		\node[draw=none,fill=none, inner sep=0.5pt] at (0,0.55) {{$s_1$}};
		\node[draw=none,fill=none, inner sep=0.5pt] at (0.37,4) {{$s_2$}};
		\end{tikzpicture}
		\caption{An example of a $2$-interval $Q$ of a funnel ($Q$ filled green and
			bounded by $s_1=\upseg(p)$ and $s_2=\lowseg(q)$).
			The red diamond vertices $a_1=p,a_2,a_3,a_4$ are the guards computed
			by Algorithm~\ref{alg:tightpath}, and $a_2,a_3$ belong to
			the interval~$Q$.
			Note that $p$ and $q$ by definition do not belong to~$Q$. 
			The shadow of $Q$ (filled light gray) is
			bounded from below by the dotted line $\lowseg(p')$, and the
			inner point $o$ on $s_1$ is the so-called {\em observer} of $Q$
			(seeing all vertices of $Q$ and, possibly, some vertices
			in the shadow). }
		\label{fig:k-interval}
	\end{figure}
	
	\begin{definition}[$t$-interval]\label{df:kinterval}\rm
		Let $F$ be a funnel with the
		and $A$ be the fixed guard set $A$ computed by Algorithm~\ref{alg:tightpath} on~$F$.
		Let $s_1$ be the base of $F$, or $s_1=\upseg(p)$ for some vertex $p$ 
		of $F$ (where $p$ is not the apex or its neighbour).
		Let $s_2$ be the apex of $F$, or
		$s_2=\lowseg(q)=\overline{qz}$ for some vertices $q$ and $z$
		of $F$ (where $q$ is not in the base of~$F$).
		Assume that $s_2$ is above $s_1$ within~$F$.
		Then the region $Q$ of $F$ bounded from below by $s_1$ and from above by
		$s_2$ is called an {\em interval of~$F$}.
		For technical reasons, $Q$ includes the vertex $z$ but
		excludes the rest of the upper boundary~$s_2$.
		Moreover, $Q$ is called a {\em$t$-interval of~$F$} if $Q$ contains at least
		$t$ of the guards of~$A$.
		See Figure~\ref{fig:k-interval} (where $p,q$ are chosen on
		the same chain, but they may also lie on opposite chains).
		
		Having an interval $Q$ of the funnel $F$, bounded from below by $s_1$,
		we define the {\em shadow of $Q$} as follows.
		If $s_1=\upseg(l_i)$ ($s_1=\upseg(r_j)$), then the shadow consists of the 
		region of $F$ between $s_1$ and $\lowseg(l_{i+1})$
		(between $s_1$ and $\lowseg(r_{j+1})$, respectively).
		If $s_1$ is the base, then the shadow is empty.
	\end{definition}
	
	\begin{lemma}\label{lem:11interval}
		If $Q$ is a $13$-interval of the funnel $F$, then there exists a point
		in $Q$ which is not visible from any vertex of $F$ outside of~$Q$.
	\end{lemma}
	\begin{proof}
		By definition, a $13$-interval has thirteen guards from~$A$ in it. By 
		the pigeon-hole principle, at least seven of these guards are on the same chain.
		Without the loss of generality, let these seven guards lie on the left chain $\mathcal{L}$ of $F$.
		Let us denote these guards by $a$, $b$, $c$, $d$, $e$, $f$ and $g$ in the bottom-up order, respectively.
		We show that the guard $d$ is not seen by any viewer outside of the $13$-interval.
		
		Suppose that $d$ can be seen by a vertex of $\mathcal{R}$ which lies below $\uptan(a) \cap \mathcal{R}$.
		This means that the vertex of $\mathcal{R}$ immediately below $\uptan(a) \cap \mathcal{R}$ 
		(say, denoted by~$x$) also sees $d$. 
		
		Since $x$ is below $\uptan(a) \cap \mathcal{R}$, $x$ sees all vertices of $\mathcal{L}$ between $b$ and $d$,
		including both $b$ and $d$.
		Additionally, if $b$ is not the immediate neighbour of $a$ on $\mathcal{L}$,
		then $x$ sees the vertex of $\mathcal{L}$ immediately below $b$ as well.
		Thus, $x$ sees all points seen by $b$ or $c$ on $\mathcal{L}$.
		
		Since $x$ sees $d$, all points of $\mathcal{R}$ that are seen by $b$ or $c$ and lie above $x$, are also seen by $d$.
		Since $x$ lies below $\uptan(a) \cap \mathcal{R}$, $a$ sees all points of $\mathcal{R}$
		between and including $x$ and $\lotan(a) \cap \mathcal{R}$.
		Since $a$ lies below $b$ and $c$ on $\mathcal{L}$, none among $b$ and $c$ can see any vertex below $\lotan(a) \cap \mathcal{R}$.
		Thus, $d$ and $a$ see all points seen by $b$ or $c$ on $\mathcal{R}$.
		
		The above arguments show that $a$, $d$ and $x$ together see all points on the two concave chains seen by $b$ and $c$.
		Observe that since $x$ is the vertex immediately below $\uptan(a) \cap \mathcal{R}$, Algorithm \ref{alg:tightpath}
		includes in $G$ either $x$, or a higher vertex $x'$ of $\mathcal{R}$ which sees the points of $F$ exclusively seen by $x$.
		This means Algorithm \ref{alg:tightpath} must choose $x$ (or, $x'$) instead of $b$ and $c$ to optimize on the number of 
		guards. Hence, we have a contradiction, and no vertex below the $13$-interval can see $d$.
		
		Now suppose that $d$ is seen by a vertex $y$ lying above the $13$-interval. Since the $13$-interval contains
		two more guards on $\mathcal{L}$ above $d$, the vertex $y$ can certainly not lie on $\mathcal{L}$.
		Therefore, $y$ must lie on $\mathcal{R}$
		If the upper segment of the $13$-interval is defined
		by $\lowseg(v)$, where $v \in \mathcal{R}$, then $d$, $e$, $f$ and $g$ must lie below
		$\lowseg(v) \cap \mathcal{L}$.
		This means, $\uptan(d) \cap \mathcal{R}$ must lie below $v$. But to see $d$, $y$ must lie below $v$.
		This makes $y$ a vertex contained in the $13$-interval. So, $v$ must lie on~$\mathcal{L}$.
		
		So, we assume that $v$ lies on $\mathcal{L}$. At the worst case, $v$ can be the guard $f$. 
		This means, $\uptan(d) \cap \mathcal{R}$ is above $lot(g) \cap \mathcal{R}$, and $y$ lies on the segment
		of $\mathcal{R}$ between these two points. But then, by an argument similar to above, $d$, $g$ and $y$ together see 
		everything that is seen by $d$, $e$, $f$ and $g$, and so Algorithm \ref{alg:tightpath} would choose only $d$, $g$
		and $y$ to get a shortest path in $G$, a contradiction. So, the point $d$ is not visible from any vertex outside of 
		the $13$-interval.
	\end{proof}

	Our second crucial ingredient is the possibility to ``almost privately'' see the vertices of an interval
	$Q$ from one point as follows.
	Recall the notation of Definition~\ref{df:kinterval}, and Figure~\ref{fig:k-interval}.
	If $s_2=\lowseg(q)$, then the intersection point of
	$\lowtan(q)$ with $s_1$ is called the {\em observer of~$Q$}.
	(Actually, to be precise, we should slightly perturb this position of the
	observer $o$ so that the visibility between $o$ and $q$ is blocked by~$z$.
	To keep the presentation simple, we neglect this detail.)
	If $s_2$ is the apex, then consider the spine of $F$ instead of $\lowtan(q)$.
	
	\begin{lemma}\label{lem:observer}
		The observer $o$ of an interval $Q$ in a funnel $F$
		can see all vertices of $Q$, but $o$ cannot see any vertex of $F$
		which is not in $Q$ and not in the shadow of~$Q$.
	\end{lemma}
	\begin{proof}
		Without the loss of generality let $q$ lie on $\mathcal{L}$.
		Since $o$ is at the intersection point of $\lowtan(q)$ with $s_1$, the observer 
		$o$ sees all vertices of the chain $\mathcal{L}$ strictly between
		its intersection with $s_1$ and~$q$ (recall that the end $p$ of
		$s_1$ may also lie on $\mathcal{R}$, but then this claim also holds).
		Likewise, $o$ sees all vertices of $\mathcal{R}$ between    
                its intersection with $s_1$ and~$z$ (other end of~$s_2$), and including $z$.

		Conversely, by the definition, $o$ is blocked from seeing
		$q$ and all vertices of $\mathcal{L}$ and $\mathcal{R}$
		above $q$ and $z$, respectively.
		If $p$ lies on the same chain of $F$ as $q$ (so $p\in\mathcal{L}$),
		as in Figure~\ref{fig:k-interval}, then $o$, which is
		slightly perturbed from $s_1=\upseg(p)$ up, cannot see $p$
		and vertices of $\mathcal{L}$ below~$p$.
		Furthermore, since $o$ is placed above the lower boundary of
		the shadow of~$Q$, it also cannot see any vertex of
		$\mathcal{R}$ below this shadow.
		The remaining case of $p\in\mathcal{R}$ is fully symmetric
		to the latter argument.
	\end{proof}
	
	The last ingredient before the main proof is the notion
	of sections of an interval $Q$ of~$F$.
	Let $s_1$ and $s_2$ form the lower and upper boundary of~$Q$.
	Consider a vertex $l_i\in\mathcal{L}$ of~$Q$.
	Then the {\em lower section of $Q$ at $l_i$} is the interval of $F$ bounded
	from below by $s_1$ and from above by $\lowseg(l_i)$.
	The {\em upper section of $Q$ at $l_i$} is the interval of $F$ bounded
	from below by $\upseg(l_i)$ and from above by $s_2$.
	Sections of $r_j\in\mathcal{R}$ are defined analogously.
	
	\begin{lemma}\label{lem:ksection}
		Let $Q$ be a $t$-interval of the funnel $F$, and let $Q_1$ and $Q_2$ be its
		lower and upper sections at some vertex $p$.
		Then $Q_i$, $i=1,2$, is a $t_i$-interval such that $t_1+t_2\geq t-3$.
		(In other words, at most $3$ of the $A$-guards in $Q$ are not in $Q_1\cup Q_2$.)
	\end{lemma}
	\begin{proof}
		The only vertices of $Q$ which are not included in $Q_1\cup Q_2$ are
		$p$ and the vertices of the shadow of~$Q_2$.
		Suppose, for a contradiction, that those contain (at least) four guards
		from~$A$; in either such case, we easily contradict minimality of the guard
		set~$A$ in Algorithm~\ref{alg:tightpath}, by the same argument given in Lemma \ref{lem:11interval}.
		The Algorithm~\ref{alg:tightpath} simply chooses $p$ and the topmost and bottommost among these (at least) four guards,
		thus choosing only three guards instead of four.
	\end{proof}
	
	Now we are ready to prove the advertised lower bound:
	\begin{theorem}\label{thm:logm-4}
		Any conflict-free chromatic guarding of a given funnel 
		requires at least $\lfloor\log_2(m+3)\rfloor-3$ colours,
		where $m$ is the minimum number of guards needed to guard the whole funnel.
	\end{theorem}
	
	\begin{proof}
		We will prove the following claim by induction on $c\geq0$:
		\begin{itemize}
			\item If $Q$ is a $t$-interval in the funnel $F$
			and $t\geq 16\cdot2^c-3$, then any conflict-free colouring of $F$
			must use at least $c+1$ colours on the vertices of $Q$ or 
			of the shadow of~$Q$.
		\end{itemize}
		
		In the base $c=0$ of the induction, we have $t\geq16-3=13$.
		By Lemma~\ref{lem:11interval}, some point of $Q$ is not seen from outside,
		and so there has to be a coloured guard in some vertex of $Q$,
		thus giving $c+1=1$ colour.
		
		Consider now $c>0$.
		The observer $o$ of $Q$ (which sees all vertices of~$Q$) 
		must see a guard $g$ of a unique colour
		where $g$ is, by Lemma~\ref{lem:observer}, 
		a vertex of $Q$ or of the shadow of~$Q$.
		In the first case, we consider $Q_1$ and $Q_2$,
		the lower and upper sections of $Q$ at~$g$.
		By Lemma~\ref{lem:ksection}, for some $i\in\{1,2\}$,
		$Q_i$ is a $t_i$-interval of $F$ such that
		$t_i\geq (t-3)/2\geq (16\cdot2^c-6)/2=16\cdot2^{c-1}-3$.
		In the second case ($g$ is in the shadow of $Q$),
		we choose $g'$ as the lowermost vertex of $Q$ on the same chain as~$g$,
		and take only the upper section $Q_1$ of $Q$ at $g'$.
		We continue as in the first case with~$i=1$.
		
		By induction assumption for $c-1$, $Q_i$ together with its shadow
		carry a set $C$ of at least $c$ colours.
		Notice that the shadow of $Q_2$ is included in~$Q$,
		and the shadow of $Q_1$ coincides with the shadow of~$Q$,
		moreover, the observer of $Q_1$ sees only a subset of the shadow of $Q$ seen
		by the observer $o$ of~$Q$.
		Since $g$ is not a point of $Q_i$ or its shadow, but our observer 
		$o$ sees the colour $c_g$ of $g$ and all colours of~$C$, we have
		$c_g\not\in C$ and hence $C\cup\{c_g\}$ has at least $c+1$ colours, as
		desired.

		Finally, we apply the above claim to $Q=F$.
		We have $t\geq m$, and for $t\geq m\geq 16\cdot2^c-3$ we derive that
		we need at least $c+1\geq\lfloor\log(m+3)\rfloor-3$ colours for guarding
		whole~$F$.
	\end{proof}

	\begin{algorithm}[tbp]
		\KwIn{A funnel $F$ with concave chains $\mathcal{L}=(l_1, l_2,
			\ldots, l_k)$ and $\mathcal{R}=(r_1, \ldots, r_m)$.}
		\KwOut{A conflict-free chromatic guard set of $F$
			using $\leq OPT+4$ colours.}
		\smallskip
	Run Algorithm~\ref{alg:tightpath} to produce 
	 a guard sequence $A=(a_1,a_2,\dots,a_t)\>$ (bottom-up)\;
	Assign colours to members of $A$ according to the ruler sequence;
         the vertex $a_i$ gets colour $c_i$
         where $c_i$ is the largest integer such that $2^{c_i}$ divides~$2i$\;		
	Output coloured guards $A$ as the (approximate) solution\;

		\caption{Approximate conflict-free chromatic guarding of a funnel\hspace*{-2ex}}
		\label{alg:apxcfreefunnel}
	\end{algorithm}

	\begin{corollary}
		\label{cor:opt+4}
		Algorithm~\ref{alg:apxcfreefunnel}, for a given funnel~$F$,
		outputs in polynomial time a conflict-free chromatic guard set $A$,
		such that the number of colours used by $A$ is by at most four larger
		than the optimum.
	\end{corollary}
	\begin{proof}
		Note the following simple property of the ruler sequence:
		if $c_i=c_j$ for some $i\not=j$, then $c_{(i+j)/2}>c_i$.
		Hence, for any $i,j$, the largest value occurring among colours
		$c_i,c_{i+1},\dots,c_{i+j-1}$ is unique.
		Since every point of $F$ sees a consecutive subsequence of~$A$,
		this is a feasible conflict-free colouring of the funnel~$F$.
		
		Let $m$ be the minimum number of guards in~$F$.
		By Lemma~\ref{lem:alg1-near-optimal}, it is $m+1\geq t=|A|\geq m$.
		To prove the approximation guarantee, observe that for $t\leq 2^c-1$,
		our sequence $A$ uses $\leq c$ colours.
		Conversely, if $t\geq 2^{c-1}$, i.e. $m\geq 2^{c-1}-1$,
		then the required number of colours for guarding $F$ is
		at least~$c-1-3=c-4$, and hence our algorithm uses at most $4$ more colours
		than the optimum.
	\end{proof}

\section{Vertex-to-Point Conflict-Free Chromatic Guarding of Weak Visibility Polygons}
	\label{sec:v2pChromatic}
	
	In this section, we extend the scope of the studied problem of vertex-to-point 
	conflict-free chromatic guarding from funnels to weak visibility polygons.
	We will establish an $O(\log^2 n)$ upper bound for the number of colours
	of vertex-guards on $n$-vertex weak-visibility polygons, and give the
	corresponding polynomial time algorithm.
	
	\begin{figure}[tbp]
		\centering
		\begin{tikzpicture}[scale=0.45]
		
		\tikzstyle{every node}=[draw, shape=circle, minimum size=4pt,inner sep=0pt];
		\tikzstyle{every path}=[thin, fill=none];
		\draw[fill=magenta!15!white,draw=none] (0,0)--(6,0.9)--(7.9,2)--(15,8.5)--(15,11.5)--(14.5,13)--(25,1.5)--(31,0)--(0,0);
		\node[fill=black,label=150:$u\,$](A) at (0,0) {};
		\node(B) at (6,0.9) {};
		\node(BB) at (7.9,2) {};
		\node(C) at (8,6) {};
		\node(D) at (7,8) {};
		\node[fill=orange](E) at (4.3,9) {};
		\node[fill=red](F) at (5,11.5) {};
		\node(G) at (7.8,9.5) {};
		\node(H) at (11,8) {};
		\node(I) at (15,8.5) {};
		\node(J) at (15,11.5) {};
		\node[fill=magenta](K) at (14.5,13) {};
		\node[fill=purple](L) at (16.5,12) {};
		\node[fill=blue](M) at (18,13) {};
		\node(N) at (19,10) {};
		\node(O) at (19.5,9) {};
		\node(P) at (23,8) {};
		\node(Q) at (24,10.5) {};
		\node[fill=cyan](R) at (24.2,12) {};
		\node[fill=green](S) at (26,11) {};
		\node[fill=green!60!black](T) at (31,10) {};
		\node(V) at (28.6,8) {};
		\node(W) at (27.5,6) {};
		\node(X) at (25,1.5) {};
		\node[fill=black,label=30:$\,v$](Y) at (31,0) {};
		\draw (A)--(B)--(BB)--(C)--(D)--(E)--(F)--(G)--(H)--(I)--(J)--(K)--(L)--(M)--(N)--(O)--(P)--(Q)--(R)--(S)--(T)--(V)--(W)--(X)--(Y)--(A);
		
		\tikzstyle{every path}=[very thick,dotted];
		\draw[orange] (A)--(B)--(BB)--(C)--(D)--(E)--(X)--(Y);
		\draw[red] (A)--(B)--(BB)--(C)--(D)--(F)--(G)--(X)--(Y);
		\draw[magenta] (A)--(B)--(BB)--(I)--(J)--(K)--(X)--(Y);
		\draw[purple] (A)--(B)--(BB)--(I)--(L)--(X)--(Y);
		\draw[blue] (A)--(B)--(BB)--(I)--(M)--(N)--(O)--(X)--(Y);
		\draw[cyan] (A)--(B)--(P)--(Q)--(R)--(X)--(Y);
		\draw[green] (A)--(B)--(P)--(S)--(X)--(Y);
		\draw[green!60!black] (A)--(B)--(T)--(V)--(W)--(X)--(Y);
		\end{tikzpicture}
		\caption{A weak visibility polygon $W$ on the base edge $uv$, and
		its collection of $8$ max funnels ordered from left to right.
		These max funnels are depicted by thick dotted lines, which
		are coloured from orange on the left to green on the right
		(the apex vertex of each max funnel is coloured the same as the funnel chains).
		Note that the left-most orange funnel is degenerate, i.e.
		having its two apex edges collinear.
		Moreover, the third (pink) max funnel from the left is
		emphasised as the filled region of~$W$.
		}
		\label{fig:maxfunnels}
	\end{figure}
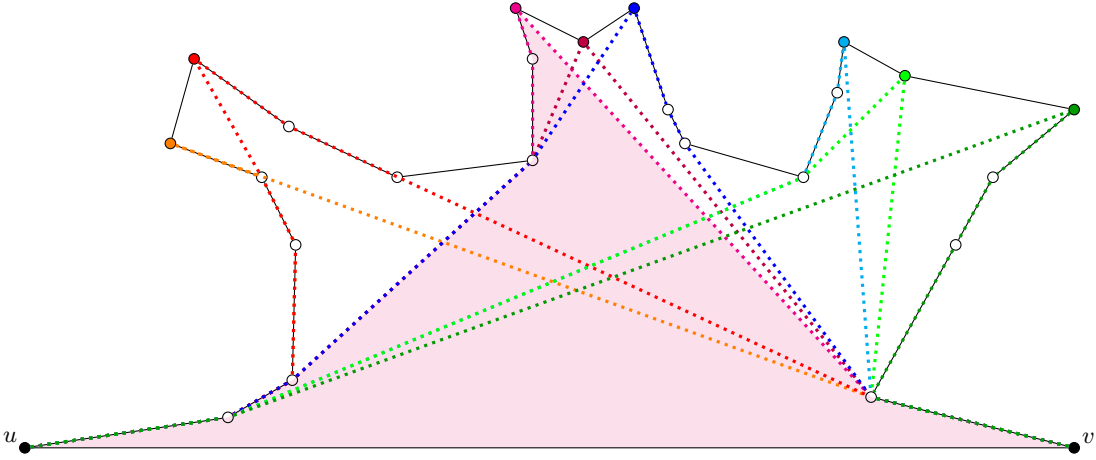

Consider a given weak visibility polygon $W$ on the base edge $uv$.
Each vertex $w$ of $W$ (other than $u,v$) is visible from some point of the segment~$\lseg uv$.
Consequently, for this $w$ one can always find a funnel $F\subseteq W$ with the base $uv$ 
and the apex~$w$ (such that $V(F)\subseteq V(W)$).
The chains of this funnel $F$ are simply the geometric shortest paths inside $W$
from $w$ to $u$ and~$v$, respectively.
The union of these funnels covers all vertices of $W$, and we consider the maximal
ones of them by inclusion -- the {\em max funnels} of $W$ -- ordered from left to right.
See Figure~\ref{fig:maxfunnels}.
Noe that, as illustrated in the picture, some max funnels may be degenerate,
which means that their two edges incident with the apex vertex are collinear.

Our algorithm (and the corresponding upper bound) 
is based on the following two rather simple observations:

\begin{itemize}
	\item In each of the max funnels, we assign colours independently in the
	logarithmic ruler sequence, similarly to Section~\ref{sec:funnelcolouring}.
	This time, however, we ``wastefully'' place guards at all vertices
	of both chains except the apex (which makes arguments easier).
	\item In order to prevent conflicts for points seeing guards from
	several funnels, we use generally different sets of colours for different
	max funnels.
	We distribute the colour sets to the max funnels again according to the
	ruler sequence, from left to right, 
	which requires only $O(\log n)$ different colour sets.
	We use the highest ruler sequence value to determine which colour set
	applies to vertices which belong to several max funnels.

	See Algorithm~\ref{alg:weakv2pcolouringB} for details.
\end{itemize}

\begin{algorithm}[tb]
	\caption{Computing a V2P conflict-free chromatic guarding of 
		a weak visibility polygon using $O(\log^2n)$ colours.}
	\label{alg:weakv2pcolouringB}
	\KwIn{A weak visibility polygon $W$, weakly visible from the base edge $uv$.}
	\KwOut{A V2P conflict-free chromatic guarding of $W$ using $O(\log^2 n)$
	colours, such that no guard is placed at any apex vertex of a max funnel of~$W$.}

	\smallskip
	Construct the (geom.) shortest path trees $T_u$ and $T_v$ of $W$ from $u$ and $v$\;
	Identify the common leaves of $T_u$ and $T_v$ as the apices of max funnels,
	and order them into a set $\ES F = \{ F_1, F_2, \ldots, F_m \}$ from left to
	right, where $m\leq n$\;
	
	For $\ell=1+\lfloor\log m\rfloor$,
	consider $\ell$ colour sets $C_1, C_2, \ldots, C_{\ell}$ of $2\log n$ colours
	each, i.e., let $C_i=C^l_i\cup C^r_i$ and all these sets be pairwise 
	disjoint of cardinality~$\log n$\;
	
	\ForEach{max funnel $F_i\in\ES F$}{
	Assign colour set $C_j$ to $F_i$, where $j$ is the largest power of $2$ dividing~$2i$%
	\tcc*{indices $j$ of the assigned sets $C_j$ form the ruler sequ.}
	}
	\ForEach{vertex $w$ of $W$}
	{
		Associate $w$ with the max funnel that was assigned colour set $C_j$ of
		the highest index~$j$ (among the max funnels containing~$w$, this is unique)\;
		\label{lin:7}
	}
	\ForEach{max funnel $F_i\in\ES F$}{
		Colour all vertices of the left chain of $F_i$ that are 
		associated with $F_i$, except the apex, by the ruler sequence 
		using the colour set~$C^l_j$ assigned to~$F_i$\;
		Colour the same way the right chain of $F_i$ using the colour set~$C^r_j$\;
	}
	\Return{Coloured $W$}\;
\end{algorithm}

Let the ordered set of max funnels (possibly degenerate) of $W$ be $\ES F = \{F_1, F_2, \ldots, F_m \}$
in the order of their apices occurring on the clockwise boundary of $W$ from $u$ to~$v$.
Then, as already mentioned, every vertex of $W$ belongs to one or more funnels of $\ES F$.
Observe that $u$ and $v$ belong to all funnels of $\ES F$.
Furthermore, for $F_i\in\ES F$ with the apex vertex~$a_i$, the point set $W\setminus F_i$
is the union of internally-disjoint polygons such that some of them are to
the {\em left of~$F_i$} (those having their vertices between $u$ and $a_i$
in the clockwise order) and the others to the {\em right of~$F_i$}.
Observe that, for $F_j\in\ES F$, we have $j<i$ if and only if $F_j\setminus F_i$
is all to the left of~$F_i$.

We have the following lemma on $W$ and $\ES F$.

\begin{lemma} \label{lem:atleastone}
Consider a weak visibility polygon $W$, and its max funnel $F_i\in\ES F$.
Assume that an observer $p\in W\setminus F_i$ is to the left of $F_i$ in
$W$, and that $x\in W\setminus F_i$ is a point of $W$ to the right of $F_i$
such that $p$ sees~$x$.
Then $p$ sees at least one vertex of $W$ belonging to $F_i$ except the apex.
\end{lemma}

Note that while it is immediate that the line of sight between $p$ and $x$
must cross both chains of the funnel $F_i$, this fact alone does not imply that
there is a vertex of $F_i$ visible from~$p$.
\begin{proof}
By the definition, $p$ sees a point $b\in\overline{uv}$ on the base of~$W$.
Consider the ray $\overrightarrow{pb}$ and rotate it counterclockwise until
it first time hits a vertex $c$ on the right chain of~$F_i$.
If $p$ sees~$c$, then we are done.
Otherwise, the line of sight $\overline{pc}$ is blocked by a point
$y\in\overline{pc}$ such that $y$ is outside of $W$, and hence to the left of $F_i$, too.
Let $L_i$ be the point set of the left chain of $F_i$.

Now, both segments $\overline{pb}$ and $\overline{px}$ cross the $L_i$, 
and the counter-clockwise order of the rays from $p$
is $\overrightarrow{pb}$, $\overrightarrow{pc}$, $\overrightarrow{px}$.
Consequently, the point set 
$\overline{pb}\cup\overline{px}\cup L_i$ (which is part of $W$)
separates~$y$ (which is in the exterior of $W$).
This contradicts the fact that $W$ is a simple polygon.
\end{proof}

Informally, Lemma~\ref{lem:atleastone} shows that the subcollection of max
funnels whose vertices are seen by the observer $p$, forms a consecutive
subsequence of $\ES F$, and so precisely one of these funnels visible by $p$
gets the highest colour set according to the ruler sequence.
From this colour set we then get the unique colour guard seen by~$p$.

The straightforward proof of correctness of our algorithm follows.

\begin{theorem}\label{thm:twoApprox}
Algorithm~\ref{alg:weakv2pcolouringB} in polynomial time computes a conflict-free
chromatic guarding of a weak visibility polygon using $O(\log^2 n)$ colours.
\end{theorem}
\begin{proof}
	Clearly, identifying the max funnels and the indexing can be done in $O(n^2)$ time. 
	We now prove the correctness of the algorithm.
	
	Any given point $p\in W$ always sees a unique colour within any single funnel $F_i\in\ES F$, 
	if at least some vertex of $F_i$ is seen from~$p$;
	in fact one such unique colour left the left or/and one from the right chain
	of this funnel (and the colour sets of the left and right chains are disjoint).
	This is because $p$ always sees a consecutive section of the left or
	right chain of $F_i$, regardless of whether $p$ is inside or outside
	of~$F_i$, and we use the ruler sequence colours on each chain.
	However, $p$ may see vertices associated with two or more max funnels assigned the same colour set
	by Algorithm~\ref{alg:weakv2pcolouringB} (line~\ref{lin:7}).
	
	Suppose the latter; that the point $p$ sees vertices associated with
	two distinct max funnels $F_i,F_j\in\ES F$ assigned the same colour set.
	Let $k$, $i<k<j$, be such that $F_k$ is assigned the highest colour
	set within the ruler subsequence from $i$ to $j$ ($k$ is unique).
	If $p\in F_k$, then $p$ sees a unique colour guard from those on~$F_k$.
	Otherwise, up to symmetry, $p$ is to the left of $F_k$ and $p$ sees
	a vertex associated with $F_j$ which is to the right of $F_k$. 
	We have a situation anticipated by Lemma~\ref{lem:atleastone}, and
	so $p$ sees some vertex of $F_k$.
	Then again, $p$ gets a unique colour guard from those on~$F_k$.
\end{proof}

\section{Vertex-to-Point Conflict-Free Chromatic Guarding of Polygons}
\label{sec:allsimple}

In this section, we extend the scope of the studied problem of
vertex-to-point conflict-free chromatic guarding from funnels and weak
visibility polygons to general simple polygons.  
We will establish an $O(\log^2 n)$ upper bound for the
number of colours of vertex-guards on $n$-vertex simple polygons, and give
the corresponding polynomial time algorithm.
For that we use the previous algorithm for chromatic guarding of weak visibility polygons
(Algorithm~\ref{alg:weakv2pcolouringB}).
Moreover, our algorithm is ready for further improvements in
Algorithm~\ref{alg:weakv2pcolouringB}, as it uses $O(C+\log n)$ colours
where $C$ is the number of colours used by Algorithm~\ref{alg:weakv2pcolouringB}.

We give our algorithm in two phases.
Namely, the decomposition phase, and the colouring phase.

\subsection{Decomposition into weak visibility polygons}

Our colouring scheme relies on the decomposition of a simple polygon into weak visibility polygons.
We utilise the decomposition algorithm described by B\"{a}rtschi \emph{et al.}
\cite{Bartschi-2014}, given in Algorithm~\ref{alg:bartdecompose}.
For an edge $f$ (or a vertex) of a polygon $P$,
we call the {\em visibility polygon of $f$} the set of all points of $P$ 
which are visible from some point of~$f$.

\begin{algorithm}[tb]
	\KwIn{A simple polygon $P$, an arbitrary vertex $v$}
	\KwOut{Decomposition of $P$ into weak visibility polygons}
	\smallskip
	$W_1 \gets$ visibility polygon of $v$; \tcc*{i.e., the set of points visible from $v$}
	$i \gets 1$, $j \gets 2$\;
	\Repeat{$\bigcup_{W \in \mathcal{W}} = P$}{
		$\mathcal{W} \gets \mathcal{W} \cup \{W_i\}$\;
		\ForEach{edge $e$ of $W_i$ which is not a boundary edge of $P$}
		{	\qquad\hfill\qquad\hfill\qquad\tcc{that is, $e$ cuts $W_i$ from the rest of $P$}
			$W_j \gets$ visibility polygon of $e$ in the
				adjacent subpolygon of $P\setminus W_i$\;
			$j \gets j + 1$\;
		}
		$i \gets i+1$\;
	}
	\Return{$\mathcal{W}$}\;
	\caption{Polygon decomposition algorithm of B\"{a}rtschi \emph{et al.} \cite{Bartschi-2014}}
	\label{alg:bartdecompose}
\end{algorithm}

The Algorithm~\ref{alg:bartdecompose} takes a simple polygon $P$ and an arbitrary vertex $v$ of $P$.
First, the algorithm computes the visibility polygon $W_1$ of $v$ and removes~$W_1$ from $P$.
Then, until the whole polygon is partitioned, the algorithm selects an edge $e$ of
a previously removed polygon, computes the visibility polygon $W_j$ of $e$ within
the rest of $P$, and then removes $W_j$, and so on.

The mentioned decomposition algorithm has been performed in
\cite{Bartschi-2014} to obtain a point-to-point guarding, 
in which the guards are not necessarily selected at vertices of the polygon.
In our case, we need to ensure that recursively chosen guards in weak
visibility subpolygons $W_j$ of $P$ are placed at vertices of $P$.
However, Algorithm~\ref{alg:bartdecompose} typically creates subpolygons
whose vertices are internal points of the edges of~$P$.
To overcome this problem, we slightly modify the algorithm by inserting an
intermediate phase -- creating a special subpolygon, which ``recovers'' the
property that the base edge of each subsequently constructed visibility polygon
is between two vertices of $P$ again.
As we will show later, this modification does not weaken the decomposition
technique much.

\begin{algorithm}[tbp]
	\KwIn{A simple polygon $P$, an arbitrary edge $e_0$ of $P$}
	\KwOut{Decomposition of $P$ into weak visibility polygons of two kinds}
	\smallskip
	$W_1 \gets$ visibility polygon of $e_0$ in $P$\;
	$i \gets 1$, $j \gets 2$\;
	\Repeat{$\bigcup_{W \in \mathcal{W}} = P$}{
		$\mathcal{W} \gets \mathcal{W} \cup \{W_i\}$\;
		\ForEach{edge $e=uv$ of $W_i$ which is not a boundary edge of $P$}{
			$U \gets$ visibility polygon of $e$ in the
				adjacent subpolygon of $P\setminus W_i$\;
			\If{both $u$ and $v$ are vertices of $P$}
			{ \label{it:uvv}
				$W_j \gets$ $U$ \tcc*{got an ordinary subpolygon $U$}
				$j \gets j + 1$\;
			}
			\Else{ \label{it:uvxy}
				Let $x$ and $y$ be the vertices of $U$ such that $x\not=v$
				is adjacent to $u$ and $y\not=u$ is adjacent to~$v$
				\tcc*{$x,y$ are vertices of $P$, too}
				$\Pi_{xy} \gets$ geometric shortest path in $U$ between $x$ and $y$\;
				$U_1$ $\gets$ the subpolygon of $U$ bounded by the
				paths $\Pi_{xy}$ and $(xu,uv,vy)$\;
				$\mathcal{W} \gets \mathcal{W} \cup \{U_1\}$
					\tcc*{adding a forward subpolygon $U_1$}
				\ForEach{edge $f$ of the path $\Pi_{xy}$}
				{ \label{it:pixy}
					$W_j \gets$ visibility polygon of $f$ in the
						adjacent subpolygon of $P\setminus U_1$\;
					\tcc{"forwarding" to ordinary subpolygons adjacent to $U_1$}
					$j \gets j + 1$
				}
			}
		}
		$i \gets i+1$\;
	}
	\Return{hierarchically structured decomposition $\mathcal{W}$ of $P$}\;
	\caption{Adjusted polygon decomposition algorithm}
	\label{alg:ourdecompose}
\end{algorithm}

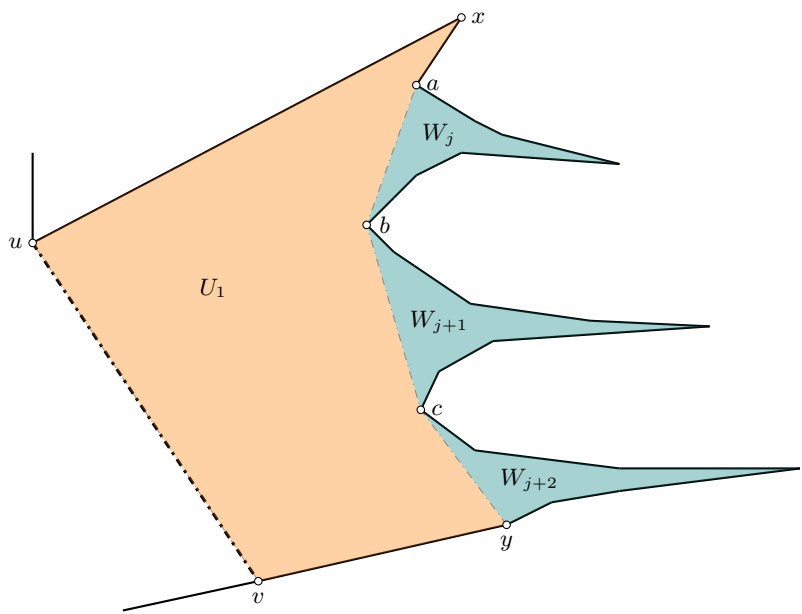
\begin{figure}[tbp]
\centering
\begin{tikzpicture}[xscale=-0.6, yscale=0.3]
\coordinate (1) at (7,-1);
\coordinate (2) at (1.5,1.5);
\coordinate (3) at (0.5,2.5);
\coordinate (4) at (-1,3);
\coordinate (5) at (-5,4);
\coordinate (6) at (-1,4);
\coordinate (7) at (2.2,4.8);
\coordinate (8) at (3.4,6.6);
\coordinate (9) at (3,8.3);
\coordinate (10) at (1.8,9.65);
\coordinate (11) at (-3,10.3);
\coordinate (12) at (-0.35,10.55);
\coordinate (13) at (2.3,11.3);
\coordinate (14) at (4,13.6);
\coordinate (15) at (4.6,14.8);
\coordinate (16) at (4,16);
\coordinate (17) at (3.5,17);
\coordinate (18) at (2.5,18);
\coordinate (19) at (-1,17.5);
\coordinate (20) at (1.6,18.8);
\coordinate (21) at (2.2,19.4);
\coordinate (22) at (3.5,21);
\coordinate (23) at (2.5,24);
\coordinate (24) at (12,14);
\foreach \i in {1,...,23}
	{
		\pgfmathtruncatemacro\j{\i+1};
		\draw[thick] (\i)--(\j);
	}
\draw[very thick,dash dot] (24)--(1);
\coordinate (A) at (10,-2.3);
\coordinate (B) at (12, 18);
\draw[thick] (A)--(1);
\draw[thick] (B)--(24);

\draw[draw=none,fill=teal, opacity=0.35] (2)--(3)--(4)--(5)--(6)--(7)--(8)--(2);
\draw[draw=none,fill=teal, opacity=0.35] (8)--(9)--(10)--(11)--(12)--(13)--(14)--(15)--(8);
\draw[draw=none,fill=teal, opacity=0.35] (15)--(16)--(17)--(18)--(19)--(20)--(21)--(22)--(15);
\draw[dash dot,fill=orange, opacity=0.35] (1)--(2)--(8)--(15)--(22)--(23)--(24)--(1);

\node at (8,12) {$U_1$};
\node at (1,3.5) {$W_{j+2}$};
\node at (3,10.5) {$W_{j+1}$};
\node at (3,18.8) {$W_j$};

\tikzstyle{every node}=[draw,fill=white, shape=circle, inner sep=1pt];
\node[label=left:$u$] at (24) {};
\node[label=below:$v$] at (1) {};
\node[label=below:$y$] at (2) {};
\node[label=right:$c$] at (8) {};
\node[label=right:$b$] at (15) {};
\node[label=right:$a$] at (22) {};
\node[label=right:$x$] at (23) {};
\end{tikzpicture}
\caption{An illustration of forward partitioning in Algorithm~\ref{alg:ourdecompose}.
	The base edge $e=uv$ of its visibility polygon $U$ (to the right) has one
	end $v$ which is not a vertex of $P$.
	In this situation we add an intermediate forward subpolygon $U_1$
	(coloured orange) bounded by $e$ and the path $\Pi_{xy}=(x,a,b,c,y)$.
	The child ordinary subpolygons of $U_1$ are then the visibility
	polygons of the edges of $\Pi_{xy}$, denoted in the picture by $W_j,W_{j+1},W_{j+2}$.
	Notice also that, in general, the union $U_1\cup W_{j}\cup
	W_{j+2}\cup\dots$ may be larger, as witnessed in the picture by~$W_j$, than
	the visibility polygon $U$ of the edge~$e$.}
\label{fig:forward}
\end{figure}

\medskip
We describe the full adjusted procedure in Algorithm~\ref{alg:ourdecompose}.
The constructed decomposition there consists of two kinds of polygons;
\begin{itemize}\item
the {\em ordinary} weak visibility polygons whose base edge has both ends
vertices of~$P$ (while some other vertices may be internal points of edges of~$P$),
\item
the {\em forward} weak visibility polygons whose base edge has at least one
end not a vertex of $P$ (and, actually, exactly one end, but this fact is not
crucial for the arguments), but all their other vertices are vertices of $P$
and form a concave chain (which will be important).
\end{itemize}
This two-sorted process is called \emph{forward partitioning} (cf.~the
branch from line~\ref{it:uvxy} of the algorithm).
We illustrate its essence in Figure~\ref{fig:forward}.
Notice that it may happen that $x=y$ and $\Pi_{xy}$ is a trivial one-vertex path.

\begin{lemma}\label{lem:ourdecompose}
Algorithm~\ref{alg:ourdecompose} runs in polynomial time, and it outputs a
decomposition $\mathcal{W}$ of $P$ into weak visibility polygons such that
the following holds:
\begin{enumerate}
\item [a)]
If $U\in\mathcal{W}$ is a forward polygon, then all vertices of $U$ except possibly
the base ones are vertices of $P$.
The non-base vertices form a concave chain.
\item [b)]
If $W\in\mathcal{W}$ is an ordinary polygon, then all vertices of $W$ 
are vertices of $P$, except possibly for apex vertices of max funnels of $W$
(cf.~Section~\ref{sec:v2pChromatic}).
\end{enumerate}
\end{lemma}

\begin{proof}
It is an easy routine to verify that the algorithm can be implemented in polynomial
time and that the members of $\mathcal{W}$ are pairwise internally disjoint
subpolygons which together cover~$P$.

Claim (a) follows trivially from the choice of $\Pi_{xy}$ as a geometric
shortest path in $U$.
For claim (b), the base vertices of $W$ are vertices of $P$ by
the condition on line \ref{it:uvv}, or the choice of base edge $f$ on line
\ref{it:pixy} of Algorithm~\ref{alg:ourdecompose}.
If some other vertex $w$ of $W$ is not a vertex of $P$, then since $W$ is a
visibility polygon of $f$, the vertex $w$ must not be reflex, and so $w$ is
the apex of some max funnel of~$W$.
\end{proof}

\subsection{Recursive colouring of the decomposition}

\begin{figure}[tbp]
\centering
\begin{tikzpicture}[xscale=1.4, yscale=0.8, rotate=-93]
\coordinate (1) at (3.6,4);
\coordinate (2) at (5.2,4.2);
\coordinate (3) at (3.5,4.7);
\coordinate (4) at (6.6,4.4);
\coordinate (5) at (4,5);
\coordinate (6) at (6.4,5);
\coordinate (7) at (3.5,6);
\coordinate (8) at (5.8,5.8);
\coordinate (9) at (4,7);
\coordinate (10) at (5.2,6.6);
\coordinate (11) at (4.4,7.6);
\coordinate (12) at (8,4);
\coordinate (13) at (10.6,6.8);
\coordinate (14) at (11,5);
\coordinate (15) at (12.5,7.5);
\coordinate (16) at (16,7);
\coordinate (17) at (15,8);
\coordinate (18) at (11,9);
\coordinate (19) at (16.2,8);
\coordinate (20) at (16.5,4);
\coordinate (21) at (15.5,6.5);
\coordinate (22) at (10,4);
\coordinate (23) at (11,3.8);
\coordinate (24) at (12,4);
\coordinate (25) at (11.5,3.5);
\coordinate (26) at (12,3);
\coordinate (27) at (10,3.6);

\foreach \i in {1,...,18,20,21,22,23,24,25,26}
	{
		\pgfmathtruncatemacro\j{\i+1};
		\draw[thick] (\i)--(\j);
	}
\draw[thick] (27)--(1);

\draw[dash dot, very thick] (19)--(20);
\coordinate (A) at (16.27,7.13);
\coordinate (B) at (16.29,6.76);
\coordinate (C) at (9.29,3.64);
\coordinate (D) at (7.31,3.77);

\draw[dash dot,fill=orange, opacity=0.3] (12)--(14)--(15)--(16)--(17)--(18)--(19)--(20)--(21)--(22)--(C)--(D)--(12);

\coordinate (P) at (5,3);
\coordinate (Q) at (17,7);

\coordinate (X) at (5,1.5);
\coordinate (Y) at (17,7.5);

\draw[thick,dotted,color=blue!50!black] (P)--(Q);
\draw[thick,dotted,color=blue!50!black] (X)--(Y);

\draw[dash dot, fill=magenta, opacity=0.25] (1)--(2)--(3)--(4)--(5)--(6)--(7)--(8)--(9)--(10)--(11)--(12)--(D)--(1);
\draw[dash dot, fill=magenta, opacity=0.5] (1)--(2)--(4)--(6)--(8)--(10)--(11)--(12)--(D)--(1);

\draw[dash dot, fill=green!60!black, opacity=0.25] (23)--(24)--(25)--(26)--(27)--(23);
\draw[dash dot, fill=green!60!black, opacity=0.5] (22)--(23)--(27)--(C)--(22);

\draw[dash dot, fill=magenta, opacity=0.25] (12)--(13)--(14)--(12);
\node at (16.8,6) {$e$};
\node at (13,6.5) {$W$};
\node at (12.5,3.5) {left child};
\node at (8,5.5) {right children};
\end{tikzpicture}
\caption{A weak visibility polygon $W$ of the base edge $e$ (coloured
orange), its left child (coloured green) and two right children (coloured pink).
The two forward polygons among all three children of $W$ are filled with
darked colour.}
\label{fig:decompositione}
\end{figure}
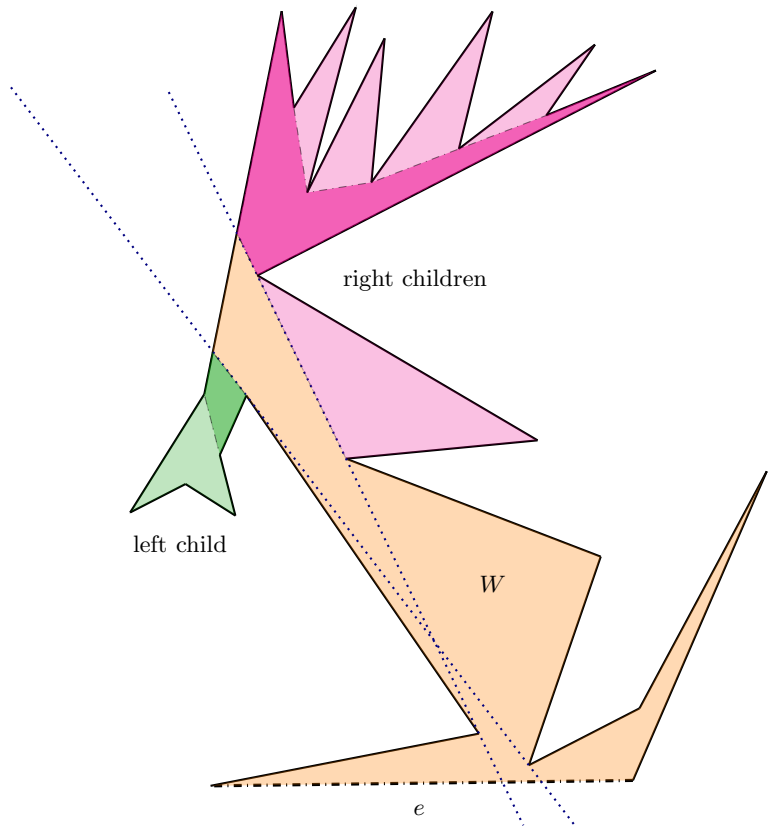

In the colouring phase, we traverse the (naturally rooted) decomposition tree of the
decomposition~$\mathcal{W}$ of $P$ computed by Algorithm~\ref{alg:ourdecompose}.
We will use the terms {\em parent} and {\em child} polygon with respect to
this rooted decomposition tree (which, essentially, is a BFS tree).
For each ordinary weak visibility polygon $W$ in $\mathcal{W}$, we apply our
Algorithm~\ref{alg:weakv2pcolouringB} for guarding, which is a correct usage since only
some apices of the max funnels of $W$ are not vertices of $P$, and those are
not used for guards.
On the other hand, for each forward weak visibility polygon $U$ in
$\mathcal{W}$, we apply a straightforward guarding by a ruler sequence on the
non-base vertices (which are vertices of $P$ and form a concave chain in~$U$).
For vertices shared between a parent and a child polygon we apply the
parental colour (which automatically resolves also sibling conflicts).

For this colouring scheme, in order to avoid colour conflicts between
different polygons, we need to use disjoint subsets of colours for the
ordinary and the forward polygons.
Moreover, this whole set of colours will be used in three disjoint copies
-- the first set of colours given to a parent polygon, the second one given
to all its ``left'' children and the third one to its ``right'' children.
Each child polygon will then reuse the other two colour sets for its children.
The reason why this works is essentially the same as in B\"{a}rtschi
\emph{et al.} \cite{Bartschi-2014}.

We need to define what the left and right children mean.
Consider an ordinary polygon $W\in\mathcal{W}$ with the base edge $e$,
and picture $W$ such that $e$ is drawn horizontal with $W$ above it.
See Figure~\ref{fig:decompositione}.
Then every edge $f$ of $W$ which is not a boundary edge of $P$ is {\em not} horizontal.
Hence every child polygon $U$ of $W$ in the decomposition (regardless of whether $U$
is a forward or ordinary polygon) is either to the left of $f$ or to the
right of it, and then $U$ is called a {\em left} or {\em right} child of $W$, respectively.
If $U$ is a left child of $W$ and $U$ is a forward polygon, then the
child ordinary polygons of $U$ are also called {\em left} children.
The same applies to right children.

We can now state the precise procedure in Algorithm~\ref{alg:coloringoverall}.

\begin{algorithm}[tb]
	\caption{Computing a V2P conflict-free chromatic guarding of 
		a simple polygon using $O(\log^2n)$ colours.}
	\label{alg:coloringoverall}
	\KwIn{A simple polygon $P$.}
	\KwOut{A V2P conflict-free chromatic guarding of $P$ using $O(C+\log n)$
	colours, where $C$ is the maximum number of colours used by calls to
	Algorithm~\ref{alg:weakv2pcolouringB}.}

	\smallskip
	Call Algorithm~\ref{alg:ourdecompose} to get the hierarchically
		structured decomposition $\mathcal{W}$ of $P$\;
	$A \gets \{1,2,\ldots,C\}$, $B \gets \{C\!+\!1,\ldots,C\!+\!\lfloor\log n\rfloor\}$
		\tcc*{colour sets to be used}
	\ForEach{ $i=1,2,3$ }{
		$C_i \gets$ $C_i'\cup C_i''$ where $C_i'$ is a disjoint
			copy of $A$ and $C_i''$ a disjoint copy of $B$\;
	}
	$W \gets$ the root polygon in the decomposition $\mathcal{W}$\;
	Call procedure RecursiveColour($W$,1)
		\tcc*{as defined below}
	\Return{Coloured $P$}\;
\end{algorithm}

\begin{procedure}[tb]
  	\caption{RecursiveColour($W$,\,$c$) }
        \If{$W$ is a forward polygon in the decomposition $\mathcal{W}$}
	{
		\ForEach{child polygon $W_i$ of $W$}
		{ Call procedure RecursiveColour($W_i$,\,$c$)\; }
		Colour the non-base vertices of $W$ by the ruler sequence
		using colours from $C_c''$\;
		\tcc{this overrides child colours from the recursive calls}
	} \Else {
		Choose $a,b$ such that $\{c,a,b\}=\{1,2,3\}$\;
		\ForEach{child polygon $W_j$ of $W$}
		{
			\If{$W_j$ is a left child of $W$}
			{ Call procedure RecursiveColour($W_j$,\,$a$)\; }
			\Else
			{ Call procedure RecursiveColour($W_j$,\,$b$)\; }
		
		}
		Call Algorithm~\ref{alg:weakv2pcolouringB} to colour-guard $W$,
			using colours from $C_c'$\;
		\tcc{this again overrides child colours from the recursive calls}
	}
 	\Return{}
\end{procedure}

\begin{theorem}\label{thm:allpolygon}
Algorithm~\ref{alg:coloringoverall} in polynomial time computes a conflict-free
chromatic guarding of an $n$-vertex simple polygon using $O(\log^2 n)$ colours.
\end{theorem}
\begin{proof}
Overall runtime analysis follows that of previous Algorithm~\ref{alg:ourdecompose}.
As noted above, the obtained colouring is valid since we only assign guards
to those vertices of the polygons in $\mathcal{W}$ which are at the same
time vertices of~$P$; this is claimed by Lemma~\ref{lem:ourdecompose} in
connection with Algorithm~\ref{alg:weakv2pcolouringB}.

It remains to prove that the resulting colouring of $P$ is conflict-free,
i.e., that every observer can see a unique colour.
We know that the constructed colouring is conflict-free within any single
polygon of $\mathcal{W}$.
Let us say that the {\em home} polygon of an observer $x\in P$ is $W\in\mathcal{W}$
such that $x\in W$ (the parent one in case of $x$ belonging to multiple polygons).
The rest will follow if we prove that no observer in $P$ 
can see another polygon which has received the same colour set as the home
polygon of~$x$.

Consider an ordinary polygon $W\in\mathcal{W}$ which is the visibility
polygon of its base edge $e$ (such as the one in Figure~\ref{fig:decompositione}), 
and let~$W_0$ be the nearest ancestor of $W$ which is also ordinary 
(i.e., $W_0$ is the parent of $W$, or the grandparent in case the parent is a forward polygon).
Let $x\in W_o$ be our observer, and let $W_1$ be an ordinary child or grandchild of $W$ 
which uses the same colour set as~$W_0$ in Algorithm~\ref{alg:coloringoverall}.
Then the line of sight between $x$ and some non-base vertex of $W_1$ must
cross the base edge $e$ of $W$, and hence be visible from~$e$.
This is a contradiction since the observed vertex of $W_1$ does not belong to~$W$.
The same argument applies symmetrically, and also in the case of a forward
home polygon.

Consider now two sibling ordinary polygons $W_1,W_2$ which receive the same colour set.
One case is that their common parent is a forward polygon $U$.
Since $W_1$ and $W_2$ are adjacent to a concave chain of $U$, an observer
$x$ may see both of them only if $x$ belongs to $U$ or some ancestor, and so
there is no conflict for~$x$.
The second case is that they have a common ordinary grandparent/parent $W$
with base edge~$e$.
Suppose that there is an observer, say $y\in W_1$, which sees a vertex $p$ of $W_2$.
Let $q$ be any vertex of $W$ ``between'' $W_1$ and $W_2$.
Then the line of sight between $e$ and $q$ must cross the line segment $yp$,
which is absurd.
Again, the same argument may also be applied to sibling forward polygons.

We have exhausted all cases.

Finally, the number of colours used by Algorithm~\ref{alg:coloringoverall}
is $3\cdot(C+\log n)$ where $C=O(\log^2 n)$ by Theorem~\ref{thm:twoApprox},
and so we have the bound $O(\log^2 n)$.
\end{proof}

\section{Vertex-to-vertex conflict-free chromatic guarding} \label{sec:v2vcfc}
	
	In the last section, we turn to the vertex-to-vertex variant
	of the guarding problem.
	As we have noted at the beginning, the V2V weak conflict-free chromatic guarding
	problem coincides with the graph conflict-free colouring problem
	\cite{k-cfc-graph} on the visibility graph of the considered polygon.
	While, for the latter graph problem, \cite{k-cfc-graph} provided
	constructions requiring an unbounded number of colours,
	it does not seem to be possible to adapt those constructions for polygon
	visibility graphs (and, actually, we propose that the conflict-free
	chromatic number is bounded in the case of polygon visibility graphs,
	see Conjecture~\ref{cj:v2v-upper3}).

\subsection{Lower bound for V2V conflict-free chromatic guarding}
	
	We start by showing in Figure~\ref{fig:v2v-1colour} 
	that one colour is not always enough.
	To improve this very simple lower bound further, we will then need 
	a more sophisticated construction.
	
	\begin{proposition}\label{pro:v2v-3colours}
		There exists a simple polygon which has no V2V conflict-free chromatic guarding with
		$2$ colours.
	\end{proposition}

	In the proof of Proposition~\ref{pro:v2v-3colours}, we elaborate on properties of the simple example 
	from Figure~\ref{fig:v2v-1colour} (right),
	and provide the construction shown in Figure~\ref{fig:v2v-2colour}; its underlying idea is that the polygon pictured on the left requires at
	least one guard to be placed on $p_1$ or $p_2$.
	Gluing four copies of that polygon to the picture on the right, we obtain
	a polygon for which two colours are not enough.
	
	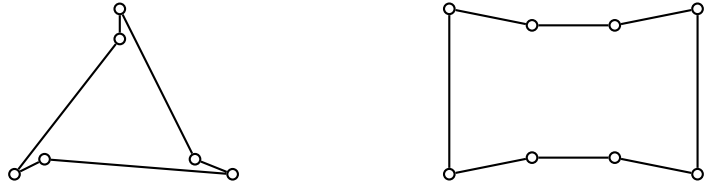
\begin{figure}[tbp]
		\centering
		\begin{tikzpicture}[scale=1]
		\tikzstyle{every node}=[draw, shape=circle, minimum size=4pt,inner sep=0pt];
		\tikzstyle{every path}=[thick];
		\node (a) at (0,0) {};
		\node (b) at (2,0) {};
		\node (c) at (1,1.6) {};
		\node (d) at (2.5,-0.2) {};
		\node (e) at (1,2) {};
		\node (f) at (-0.4,-0.2) {};
		\draw (a) -- (d) -- (b) -- (e) -- (c) -- (f) -- (a);
		\end{tikzpicture}
	\qquad\qquad\qquad\qquad
		\begin{tikzpicture}[scale=1.1]
		\tikzstyle{every node}=[draw, shape=circle, minimum size=4pt,inner sep=0pt];
		\tikzstyle{every path}=[thick];
		\node (a) at (0,0) {};
		\node (b) at (1,0.2) {};
		\node (c) at (2,0.2) {};
		\node (d) at (3,0) {};
		\node (e) at (3,2) {};
		\node (f) at (2,1.8) {};
		\node (g) at (1,1.8) {};
		\node (h) at (0,2) {};
		\draw (a) -- (b) -- (c) -- (d) -- (e) -- (f) -- (g) -- (h) -- (a);
		\end{tikzpicture}
		\caption{Two examples of simple polygons requiring at least $2$ colours for
			a V2V conflict-free chromatic guarding (in each, one vertex guard cannot see
			all vertices, and any two guards of the same colour make a conflict).}
		\label{fig:v2v-1colour}
	\end{figure}
	
	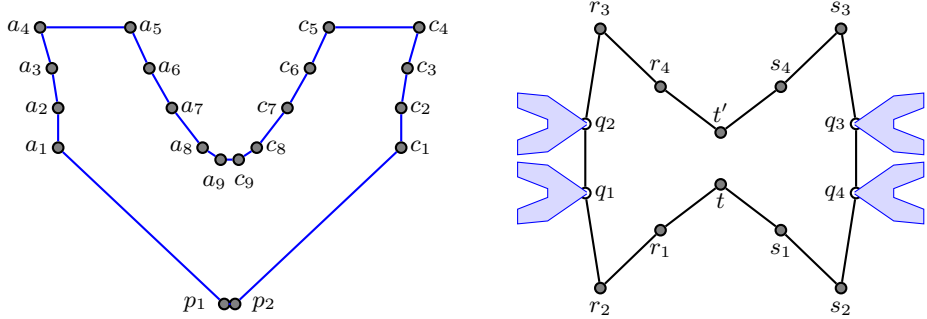
\begin{figure}[tbp]
		\centering
		\begin{tikzpicture}[xscale=1.2, yscale=1.6]
		\tikzstyle{every node}=[draw, shape=circle, minimum size=4pt,
		fill=gray, inner sep=0pt];
		\tikzstyle{every path}=[thick];
		\node[label=right:$~p_2$] (a) at (0.06,-0.3) {};
		\node[label=left:$p_1~$] (b) at (-0.06,-0.3) {};
		\node[label=left:$a_1$] (c) at (-1.9,1) {};
		\node[label=left:$a_2$] (cc) at (-1.9,1.33) {};
		\node[label=left:$a_3$] (dd) at (-1.97,1.66) {};
		\node[label=left:$a_4$] (d) at (-2.1,2) {};
		\node[label=right:$a_5$] (e) at (-1.1,2) {};
		\node[label=right:$a_6$] (ee) at (-0.89,1.66) {};
		\node[label=right:$a_7$] (ff) at (-0.64,1.33) {};
		\node[label=left:$a_8$] (f) at (-0.3,1) {};
		\node[label=below:$\!a_9~$] (g) at (-0.1,0.9) {};
		\node[label=below:$~c_9\!$] (h) at (0.1,0.9) {};
		\node[label=right:$c_8$] (i) at (0.3,1) {};
		\node[label=left:$c_7$] (ii) at (0.64,1.33) {};
		\node[label=left:$c_6$] (jj) at (0.89,1.66) {};
		\node[label=left:$c_5$] (j) at (1.1,2) {};
		\node[label=right:$c_4$] (k) at (2.1,2) {};
		\node[label=right:$c_3$] (kk) at (1.97,1.66) {};
		\node[label=right:$c_2$] (ll) at (1.9,1.33) {};
		\node[label=right:$c_1$] (l) at (1.9,1) {};
		\draw[color=blue] (a) -- (b) -- (c) -- (cc) -- (dd) -- (d)
		-- (e) -- (ee) -- (ff) -- (f) -- (g) -- (h)
		-- (i) -- (ii) -- (jj) -- (j)
		-- (k) -- (kk) -- (ll) -- (l) -- (a);
		\end{tikzpicture}
		\qquad
		\def\bowlsh#1#2{%
			\begin{scope}[shift={#1}, rotate=#2, yscale=1, xscale=0.17]
				\draw[color=blue, thin, fill=blue!15!white] (-0.06,-0.3) -- (-1.9,1) -- (-2.1,2) -- (-1.1,2)
				-- (-0.3,1) -- (0.3,1) -- (1.1,2) -- (2.1,2) -- (1.9,1) -- (0.06,-0.3);
			\end{scope}
		}
		\begin{tikzpicture}[xscale=0.4,yscale=1.15]
		\tikzstyle{every node}=[draw, shape=circle, minimum size=4pt,
		fill=gray, inner sep=0pt];
		\tikzstyle{every path}=[thick];
		\node[label=below:$t$] (a) at (0,-0.3) {};
		\node[label=below:$r_1$] (c) at (-2,-0.83) {};
		\node[label=below:$r_2$] (e) at (-4,-1.5) {};
		\node[fill=white, label=right:$q_1$] (ee) at (-4.5,-0.4) {};
		\node[fill=white, label=right:$q_2$] (ff) at (-4.5,0.4) {};
		\node[label=above:$r_3$] (f) at (-4,1.5) {};
		\node[label=above:$r_4$] (h) at (-2,0.83) {};
		\node[label=above:$t'$] (aa) at (0,0.3) {};
		\node[label=above:$s_4$] (k) at (2,0.83) {};
		\node[label=above:$s_3$] (m) at (4,1.5) {};
		\node[fill=white, label=left:$q_3$] (mm) at (4.5,0.4) {};
		\node[fill=white, label=left:$q_4$] (nn) at (4.5,-0.4) {};
		\node[label=below:$s_2$] (n) at (4,-1.5) {};
		\node[label=below:$s_1$] (p) at (2,-0.83) {};
		\draw (a) -- (c) -- (e) -- (ee) -- (ff) -- (f)
		-- (h) -- (aa) -- (k) -- (m) -- (mm)
		-- (nn) -- (n) -- (p) -- (a);
		\bowlsh{(-4.75,-0.4)}{90}
		\bowlsh{(-4.75,0.4)}{90}
		\bowlsh{(4.75,-0.4)}{270}
		\bowlsh{(4.75,0.4)}{270}
		\end{tikzpicture}
		
		\caption{A construction of an example requiring at least $3$ colours for
			a V2V conflict-free chromatic guarding (cf.~Proposition~\ref{pro:v2v-3colours}).
			The bowl shape (see on the left), suitably squeezed and with a tiny
			opening between $p_1$ and $p_2$, is placed to four
			positions within the bowtie shape on the right.}
		\label{fig:v2v-2colour}
	\end{figure}

	\begin{proof}
		We call a ``{\em bowl\/}'' the simple polygon depicted in Figure~\ref{fig:v2v-2colour} 
		(note that it contains two copies of the shape from Fig.~\ref{fig:v2v-1colour}).
		In particular, the vertices $p_1,p_2$ see all other vertices of the bowl.
		We claim the following:
		\begin{itemize}
			\item[(i)] In any conflict-free $2$-colouring of the bowl there is a guard
			placed on $p_1$ or $p_2$ (or both).
		\end{itemize}
		
		Assume the contrary to (i), that is, existence of a $2$-colouring
		of the bowl avoiding both $p_1$ and~$p_2$.
		One can easily check that the subset $A=\{a_1,\ldots,a_8\}$ of the vertices
		requires both colours, with guards possibly placed at $A\cup\{a_9\}$.
		Symmetrically, there should be guards of both colours placed at the 
		disjoint subset $\{c_1,\ldots,c_8,c_9\}$.
		Then we have got a colouring conflict at both $p_1$ and $p_2$, thus proving (i).
		(On the other hand, placing a single guard at either $p_1$ or $p_2$
		gives a feasible conflict-free colouring of the bowl.)
		
		\smallskip
		The next step is to arrange four copies of the bowl within a suitable
		simple polygon, as depicted on the right hand
		side of Figure~\ref{fig:v2v-2colour}.
		More precisely, let $S$ be the (bowtie shaped) polygon on the right of the picture.
		Note that the chains $C_1=(r_2,r_1,t,s_1,s_2)$ and
		$C_2=(r_3,r_4,t',s_4,s_3)$ of $S$ are both concave, and each vertex
		of $C_1$ sees all of~$C_2$.
		We construct a polygon $S'$ from $S$ by making a tiny opening at each of the
		vertices $q_1,q_2,q_3,q_4$, and gluing there a suitably rotated and squeezed copy of the
		bowl, where gluing is done along a copy of the edge $\lseg{p_1}{p_2}$ of the bowl.
		We call these openings at former vertices of $S$ the {\em doors} of~$S'$.
		Obviously, the doors can be made so tiny that there is no accidental
		visibility between a vertex inside a bowl and a vertex belonging
		to the rest of~$S'$.
		
		Assume, for a contradiction, that $S'$ admits a conflict-free colouring with
		two colours, say red and blue.
		Up to symmetry, let the unique colour seen by vertex $t$ be blue.
		Since $t$ sees all four doors, and (i) every door has a guard,
		at least three guards at the doors are red.
		Hence the unique colour seen by symmetric $t'$ must also be blue.
		Consequently, either there is only one blue guard at one of
		$r_1,r_4,s_1,s_4,q_1,q_2,q_3,q_4$ (which all see both $t$ and $t'$), 
		or there are two blue guards suitably placed at a pair of vertices from
		$r_2,r_3,s_2,s_3$ (each of those sees one of~$t,t'$).
		Moreover, a single blue guard cannot be placed at one of the doors
		$q_1,q_2,q_3,q_4$, because that would leave $r_3$ or $s_3$ unguarded
		(seeing two red and no blue).
		Consequently, the guards placed at the doors          
		$q_1,q_2,q_3,q_4$ must all be red, and so all vertices
		$t,r_1,r_2,r_3,r_4,t',s_1,s_2,s_3,s_4$ must be guarded by a blue guard.
		The latter is clearly impossible without a conflict
		(similarly as in Fig.~\ref{fig:v2v-1colour}).
	\end{proof}

	Our investigation of the V2V chromatic guarding problem,
	although not giving further rigorous claims (yet), moreover
	suggests the following conjecture:%
	\begin{conjecture}\label{cj:v2v-upper3}
		Every simple polygon admits a weak conflict-free vertex-to-vertex 
		chromatic guarding with at most~$3$ colours.
	\end{conjecture}

\subsection{Hardness of V2V conflict-free chromatic guarding}

	In view of Conjecture~\ref{cj:v2v-upper3} and the algorithmic results for
	other variants of chromatic guarding, it is natural to ask how difficult is
	to decide whether using $1$ or $2$ colours in V2V guarding is enough for 
	a given simple polygon.
	Actually, for general graphs the question whether one can find a
	conflict-free colouring with~$1$ colour was investigated already long time
	ago \cite{biggs-1973} (under the name of a perfect code in a graph), 
	and its NP-completeness was shown by Kratochv{\'{\i}}l and K\v{r}iv{\'a}nek
	in~\cite{DBLP:conf/mfcs/KratochvilK88}.
	Previous lower bounds and hardness results in this area, 
	including recent~\cite{k-cfc-graph}, however, do not seem to
	directly help in the case of polygon visibility graphs.
	
	Here we show that in both cases of $1$ or $2$ guard colours,
	the conflict-free chromatic guarding problem on arbitrary simple polygons is NP-complete.
	In each case we use a routine reduction from SAT, using a ``reflection model'' of a
	formula which we have introduced recently
	in~\cite{DBLP:conf/fsttcs/CagiriciHR17} for showing hardness of the ordinary
	chromatic number problem on polygon visibility graphs.
	This technique is, on a high approximate level, shown in
	Figure~\ref{fig:hardness-rough}:
	the variable values (T or F) are encoded in purely local {\em variable
		gadgets}, which are privately observed by opposite {\em reflection (or copy)
		gadgets} modelling the literals, and then reflected within precisely
	adjustable narrow beams to again opposite {\em clause gadgets}.
	Furthermore, there is possibly a special {\em guard-fix gadget} whose purpose 
	is to uniquely guard the polygonal skeleton of the whole construction,
	and so to prevent interference of the skeleton with the other local gadgets
	(which are otherwise ``hidden'' from each other).
	
	\begin{figure}[tb]
		\small\centering
		\begin{tikzpicture}[scale=2.1]
		\tikzstyle{every node}=[minimum size=0pt,inner sep=0pt];
		\draw (2.4,-0.25) node {reflection gadgets (literals)} ;
		\draw (1.25,2.2) node {clause gadgets} ;
		\draw (3.4,2.2) node {variable gadgets} ;
		\tikzstyle{every node}=[draw, fill=white, shape=circle, minimum size=2.5pt,inner sep=0pt];
		\node (cf) at (4.53,1) {} ;
		\tikzstyle{every path}=[draw,color=black];
		\draw (0.4,-0.05) arc (98:82:14) -- (cf) ;
		\draw (0,2) arc (-101:-78.5:12) -- (cf) ;
		\draw (0.4,-0.05) -- (0,2) ;
		\node[label=right:~(guard-fix)] at (4.54,1.05) {} ;
		\node[label=left:(guard-fix)~~] at (0.188,1.05) {} ;
		\node (cf0) at (0.193,1) {} ;
		\tikzstyle{every node}=[draw, fill=white, shape=circle, minimum size=17pt, inner sep=1pt];
		\node (cx1) at (0.8,0.01) {\small ${x_1}$};
		\node (nx1) at (1.12,0.03) {\small ${\!\neg x_1\!}$};
		\node (cx2) at (2.0,0.09) {\small ${x_2}$};
		\node (cx3) at (2.5,0.1) {\small ${x_3}$};
		\node (nx3) at (2.8,0.09) {\small ${\!\neg x_3\!}$};
		\node (cx4) at (3.3,0.06) {\small ${x_4}$};
		\node (cx5) at (3.8,0.02) {\small ${x_5}$};
		\tikzstyle{every node}=[draw, fill=white, shape=circle, minimum size=14pt, inner sep=1pt];
		\node (x1) at (3.1,1.8) {${x_1}$};
		\node (x2) at (3.4,1.82) {${x_2}$};
		\node (x3) at (3.7,1.85) {${x_3}$};
		\node (x4) at (4.0,1.89) {${x_4}$};
		\node (x5) at (4.3,1.94) {${x_5}$};
		\node (x6) at (4.6,1.99) {${x_6}$};
		\tikzstyle{every node}=[draw, fill=white, shape=circle, inner sep=2pt];
		\node (cl1) at (0.2,1.95) {$\substack{x_1\vee x_2\\[2pt] \vee\neg x_3}$};
		\node (cl2) at (1.0,1.85) {$\substack{\neg x_1\vee \\[2pt] x_3\vee x_4}$};
		\node (cl3) at (1.9,1.78) {~~\dots~~};
		\tikzstyle{every path}=[draw, color=red, dashed];
		\draw (cx1) -- (x1) -- (nx1) -- (cl2) ;
		\draw (cx1) -- (cl1) -- (cx2) ;
		\draw (cx2) -- (x2) ;
		\draw (cx3) -- (x3) ;
		\draw (cx3) -- (cl2) -- (cx4) ;
		\draw (cx3) -- (x3) -- (nx3) -- (cl1) ;
		\draw (cx4) -- (x4) ;
		\draw (cl3) -- (cx5) -- (x5) ;
		\tikzstyle{every path}=[draw, color=gray, dotted];
		\draw (2.5,1.5) -- (2.5,2.25);
		\end{tikzpicture}
		
		\caption{A scheme of a ``reflection'' polygonal visibility model 
			of a 3-SAT formula, as taken from \cite{DBLP:conf/fsttcs/CagiriciHR17}.
			The red dashed lines show ``visibility communication'' between
			related variable and reflection gadgets (the literals),
			and between related reflection and clause gadgets.}
		\label{fig:hardness-rough}
	\end{figure}
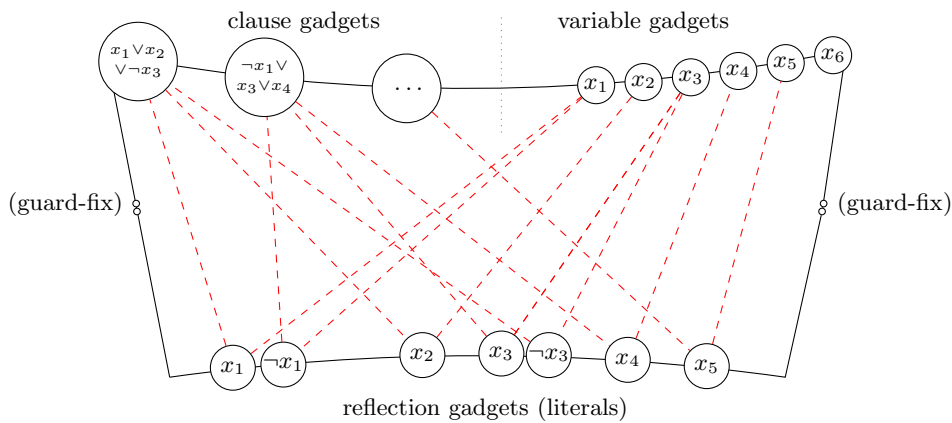

	\begin{theorem}\label{thm:v2vhardness}
		For $c\in\{1,2\}$, the question whether a given simple polygon 
		admits a weak conflict-free vertex-to-vertex chromatic guarding 
		with at most~$c$ colours, is NP-complete.
	\end{theorem}
	
	\begin{proof}
	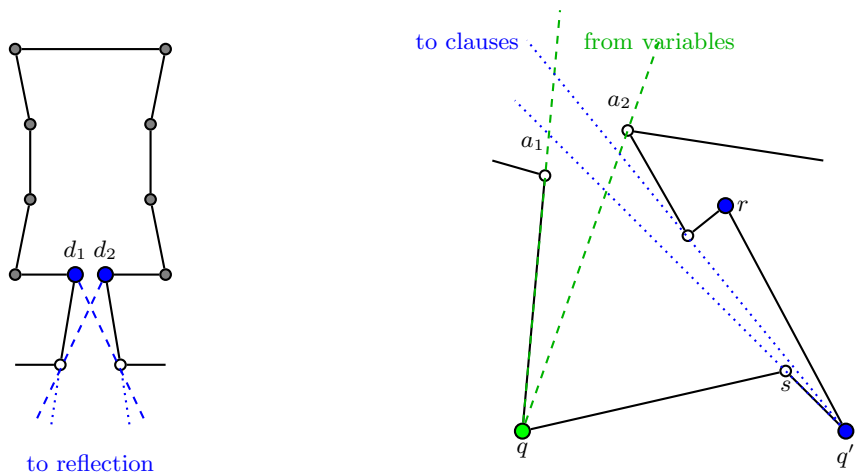
\begin{figure}[tbp]
		\centering\hfill
		\begin{tikzpicture}[scale=1]
		\tikzstyle{every path}=[thick];
		\tikzstyle{every node}=[draw, shape=circle, minimum size=4pt,
		inner sep=0pt, fill=gray];
		\node (a) at (0,0) {};
		\node (b) at (0.2,1) {};
		\node (c) at (0.2,2) {};
		\node (d) at (0,3) {};
		\node (e) at (2,3) {};
		\node (f) at (1.8,2) {};
		\node (g) at (1.8,1) {};
		\node (h) at (2,0) {};
		\tikzstyle{every node}=[draw, shape=circle, minimum size=4pt,
		inner sep=0pt, fill=none];
		\node[fill=blue,inner sep=2pt, label=above:$d_1$] (d1) at (0.8,0) {};
		\node[fill=blue,inner sep=2pt, label=above:$d_2$] (d2) at (1.2,0) {};
		\node (v1) at (0.6,-1.2) {};
		\node (v2) at (1.4,-1.2) {};
		\draw (a) -- (b) -- (c) -- (d) -- (e) -- (f) -- (g) -- (h);
		\draw (a) -- (d1) --  (v1) -- (0,-1.2) node[draw=none] {} ;
		\draw (h) -- (d2) --  (v2) -- (2,-1.2) node[draw=none] {} ;
		\draw[dotted,color=blue] (v1) -- (0.48,-2);
		\draw[dotted,color=blue] (v2) -- (1.52,-2);
		\draw[dashed,color=blue] (d2) -- (0.25,-2);
		\draw[dashed,color=blue] (d1) -- (1.75,-2);
		\node[draw=none,fill=none,color=blue, shape=rectangle] at (1,-2.5) {to	reflection};
		\end{tikzpicture}
		\qquad\hfill
		\begin{tikzpicture}[scale=2]
		\tikzstyle{every node}=[draw, shape=circle, minimum size=4pt,inner sep=0pt];
		\tikzstyle{every path}=[thick];
		\node[label=$a_1~~~$] (a) at (0.15,1.7) {};
		\node[fill=green,inner sep=2pt, label=below:$q$] (b) at (0,0) {};
		\node[label=below:$s$] (c) at (1.75,0.4) {};
		\node[fill=blue,inner sep=2pt, label=below:$q'$] (d) at (2.15,0) {};
		\node[fill=blue,inner sep=2pt, label=right:$r$] (e) at (1.35,1.5) {};
		\node (f) at (1.1,1.3) {};
		\node[label=$a_2~~$] (g) at (0.7,2) {};
		\draw (a) -- (b) -- (c) -- (d) -- (e) -- (f) -- (g) ;
		\draw (a) -- (-0.2,1.8) node[draw=none] {} ;
		\draw (g) -- (2,1.8) node[draw=none] {} ;
		\draw[dashed, color=green!70!black] (b) -- (0.25,2.8);
		\draw[dashed, color=green!70!black] (b) -- (0.91,2.6) node[draw=none] {from variables};
		\draw[dotted, color=blue] (d) -- (-0.05,2.2);
		\draw[dotted, color=blue] (d) -- (0.03,2.6) node[draw=none, label=left:to clauses] {};
		\end{tikzpicture}
		\hfill~%
		\caption{Left: the variable gadget, such that exactly one of $d_1,d_2$ must
			have a guard. ~Right: the reflection gadget, where $a_1,a_2$ are
			seen by a guard from outside, and one of $r,q'$ has a guard.}
		\label{fig:v2v-1hardness}
	\end{figure}

		For $c=1$ we reduce from the 1-in-3 SAT variant, which asks for an assignment
		having exactly one true literal in each clause and which is NP-complete, too.
		Note that in this case, there is only one colour of guards, and so every
		vertex must see exactly one guard.
		Following the general scheme of Figure~\ref{fig:hardness-rough}, we now give
		the particular variable and reflection gadgets:
		\begin{itemize}
			\item The variable gadget is depicted in Figure~\ref{fig:v2v-1hardness} on the left. 
			The opening between $d_1$ and $d_2$ is sufficiently small to
			prevent accidental visibility between an inner vertex of this gadget
			(gray in the picture) and other vertices outside.
			
			In any conflict-free $1$-colouring of a polygon containing this
			gadget, precisely one of $d_1,d_2$ must have a guard (otherwise, the inner
			vertices cannot be guarded, as in Figure~\ref{fig:v2v-1colour}),
			and there is no other guard on this gadget.
			
			\item The reflection gadget is depicted in Figure~\ref{fig:v2v-1hardness} on the right.
			Note that the visible angles of the vertices $q$ and $q'$ can be made
			arbitrarily tiny and fine-adjusted (independently of each other) by changing
			horizontal positions of $q,q'$ and~$s$.
			
			In any conflict-free $1$-colouring of a polygon containing this
			gadget, such that both $a_1,a_2$ see a guard (from outside),
			the following holds: a guard can be placed only at $q'$ or $r$,
			and the guard is at $q'$ if and only if $q$ sees a guard from outside
			(otherwise, there would be a conflict at $q$ or $q$ would not be guarded).
			
			\item Every clause gadget is just a single vertex positioned on a concave
			chain as shown in Figure~\ref{fig:hardness-rough}.
			There is no guard-fix gadget present in this construction.
		\end{itemize}
		
		For a given 3-SAT formula $\Phi=(x_i\vee\neg x_j\vee x_k)\wedge\ldots$,
		the construction is completed as follows.
		Within the frame of Figure~\ref{fig:hardness-rough}, we place a copy of the
		variable gadget for each variable of $\Phi$.
		We adjust these gadgets such that the combined visible angles of $d_1,d_2$
		of each variable gadget do not overlap with those of other variables on the
		bottom base of the frame.
		(The lower left and right corners of the frame are seen by the first and
		last variable gadgets, respectively.)
		
		For each literal $\ell$ containing a variable $x_i$ we place a copy of the
		reflection gadget at the bottom of the frame, such that its opening
		$a_1,a_2$ is visible from both $d_1$ and $d_2$ of the gadget of $x_i$.
		We adjust the narrow visible angle of $q$ (by moving $q$
		horizontally) such that $q$ sees $d_1$ but not $d_2$
		if the literal $\ell$ is $x_i$, and $q$ sees $d_2$ but not $d_1$
		if $\ell=\neg x_i$.
		Then we adjust the visible angle of $q'$ such that it sees exactly the one
		vertex (on the top concave chain of the frame) which represents the clause
		containing~$\ell$.
		
		This whole construction can clearly be done in polynomial time and 
		precision (cf.\ the similar arguments in \cite{DBLP:conf/fsttcs/CagiriciHR17}).
		To recapitulate, each of the vertices of our constructed polygon is
		\begin{itemize}
			\item coming from a copy of the variable gadget for each variable of $\Phi$, or
			\item from a copy of the reflection gadget added for each literal in~$\Phi$, or
			\item is a singleton representative of one of the clauses of $\Phi$, or
			\item is the auxiliary lower-left or right corner.
		\end{itemize}
		
		Assume that $\Phi$ has a 1-in-3 satisfying assignment.
		Then we place a guard at $d_1$ of a variable gadget of $x_i$ if $x_i$ is
		true, and at $d_2$ otherwise.
		Moreover, at a reflection gadget of a literal $\ell$, we place a guard at
		$q'$ if $\ell$ is true, and at $r$ otherwise.
		No other guards are placed.
		In this arrangement, every vertex of a variable or reflection gadget sees
		precisely one guard, regardless of the evaluation of $\Phi$.
		Moreover, since the assignment of $\Phi$ makes precisely one literal of each
		clause true, every clause vertex also sees one guard.
		This is a valid conflict-free $1$-colouring.
		
		Conversely, assume we have a conflict-free $1$-colouring of our polygon.
		Since precisely one of $d_1,d_2$ of each variable gadget of $x_i$ has a guard,
		this correctly encodes the truth value of $x_i$ (true iff the guard is
		at~$d_1$), and we know that both $a_1,a_2$ of each reflection gadget
		of a literal $\ell$ see an outside guard.
		Hence, by our construction, the gadget of $\ell$ has a guard at $q'$ if and
		only if $\ell$ is true in our derived assignment of variables.
		
		Furthermore, any clause vertex can see a guard only at a vertex $q'$ of some
		reflection gadget. 
		Otherwise (including the case of a guard at some clause vertex), we would
		necessarily get a guard conflict at a vertex $a_1$ of some reflection gadget.
		Consequently, every clause contains precisely one true literal
		(as determined by the guard visible from this clause vertex).
		The NP-completeness reduction for $c=1$ is finished.
		
		\medskip
		
        \begin{figure}[tbp]
                \def\bowlsh#1#2{%
                        \begin{scope}[shift={#1}, rotate=#2, yscale=0.4, xscale=0.17]
                                \begin{scope}[shift={(0,0.3)}]
                                        \draw[color=blue, thin, fill=blue!15!white] (-0.06,-0.3) -- (-1.9,1) -- (-2.1,2) -- (-1.1,2)
                                        -- (-0.3,1) -- (0.3,1) -- (1.1,2) -- (2.1,2) -- (1.9,1) -- (0.06,-0.3);
                                \end{scope}
                        \end{scope}
                }
                \centering
                \begin{tikzpicture}[scale=1.4]
                \tikzstyle{every node}=[draw, shape=circle, minimum size=3pt,
                	inner sep=0pt, fill=gray];
                \tikzstyle{every path}=[thick];
                \node (a) at (0,0) {};
                \node (b) at (4,0) {};
                \node (c1) at (4.1,1) {};
                \node[label=below:$c_4$] (c2) at (4.4,1.1) {};
                \node[fill=white, minimum size=4pt, label=below:$d_4$] (c3) at (5.5,1.1) {};
                \node[fill=white, minimum size=4pt, label=above:$d_3$] (c4) at (5.5,1.9) {};
                \node[label=above:$c_3$] (c5) at (4.4,1.9) {};
                \node (c6) at (4.15,2) {};
                \node[fill=white, minimum size=4pt, label=below:$~~~d_2$] (d1) at (4.15,4.1) {};
                \node[fill=white, minimum size=4pt, label=above:$~~~d_1$] (d2) at (4.1,4.9) {};
                \node (e) at (4,6) {};
                \node (h) at (0,6) {};
                \node[label=$a_1~~~$] (j1) at (-0.1,4.9) {};
                \node (j2) at (-1.3,4.8) {};
                \node (j3) at (-1.5,4.72) {};
                \node (j4) at (-1.7,4.66) {};
                \node[label=$~~~A~~$] (j5) at (-1.9,4.63) {};
                \node (j6) at (-1.9,4.37) {};
                \node (j7) at (-1.7,4.34) {};
                \node (j8) at (-1.5,4.28) {};
                \node (j9) at (-1.3,4.2) {};
                \node[label=below:$\!\!a_2~~$] (j10) at (-0.15,4.1) {};
                \node[label=$b_1~~~$] (k1) at (-0.15,1.9) {};
                \node (k2) at (-1.3,1.8) {};
                \node (k3) at (-1.5,1.72) {};
                \node (k4) at (-1.7,1.66) {};
                \node[label=$~~~B~~$] (k5) at (-1.9,1.63) {};
                \node (k6) at (-1.9,1.37) {};
                \node (k7) at (-1.7,1.34) {};
                \node (k8) at (-1.5,1.28) {};
                \node (k9) at (-1.3,1.2) {};
                \node[label=below:$\!\!b_2~~$] (k10) at (-0.1,1.1) {};
                \draw (b) -- (c1) -- (c2) -- (c3) -- (c4) -- (c5)
                -- (c6) -- (d1) -- (d2) -- (e);
                \draw (h) -- (j1) -- (j2) -- (j3) -- (j4) -- (j5) 
                -- (j6) -- (j7) -- (j8) -- (j9) -- (j10) 
                -- (k1) -- (k2) -- (k3) -- (k4) -- (k5) 
                -- (k6) -- (k7) -- (k8) -- (k9) -- (k10) -- (a);
                \draw[dashed] (a) arc (98:82:14.2) -- (b) ;
                \draw[dashed] (h) arc (-98:-82:14.2) -- (e) ;
                \bowlsh{(c3)}{270}
                \bowlsh{(c4)}{270}
                \bowlsh{(d1)}{270}
                \bowlsh{(d2)}{270}
                \tikzstyle{every path}=[thick];
                \draw[dotted, color=blue] (5.5,1.13) -- (k9);
                \draw[dotted, color=blue] (5.5,1.13) -- (k6);
                \draw[dotted, color=blue] (5.5,1.87) -- (k5);
                \draw[dotted, color=blue] (5.5,1.87) -- (k2);
                \draw[dotted, color=gray] (c2) -- (k10);
                \draw[dotted, color=gray] (c5) -- (k1);
                \draw[dotted, color=blue] (d1) -- (j9);
                \draw[dotted, color=blue] (d1) -- (j6);
                \draw[dotted, color=blue] (d2) -- (j5);
                \draw[dotted, color=blue] (d2) -- (j2);
                \end{tikzpicture}
                
                \caption{The guard-fix gadget for conflict-free $2$-colouring:
                        the four blue filled shapes glued to the frame on the right are copies of the bowl
                        shape from Figure~\ref{fig:v2v-2colour}, and the two ``pockets'' 
                        on the left side of the frame enforce each pair of bowls to receive guards of both
                        colours. The middle part of the frame is much wider than depicted here.
                        See in the proof of Theorem~\ref{thm:v2vhardness}.}
                \label{fig:v2v-2colourF}
        \end{figure}
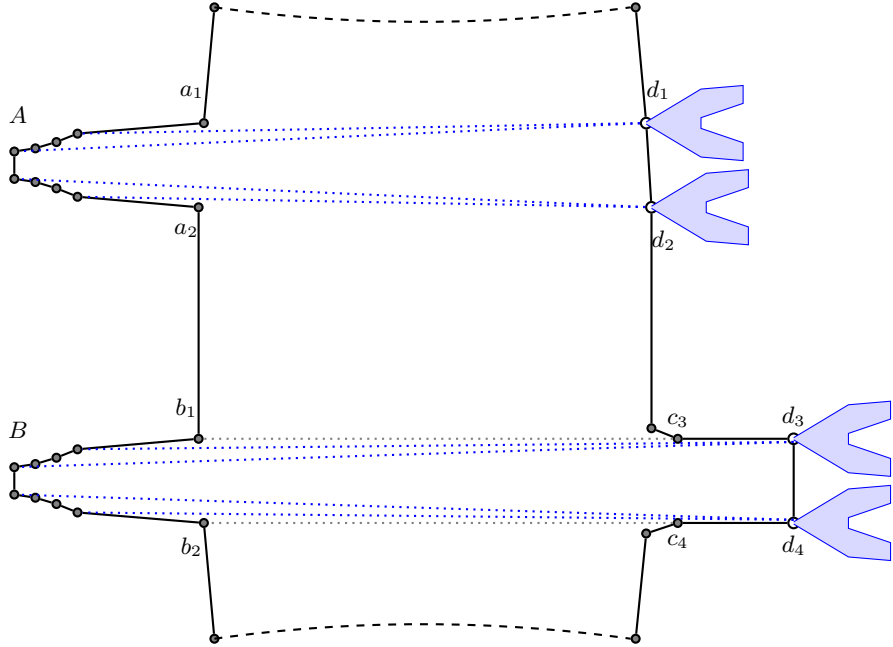

		We now move onto the $c=2$ case, which we reduce from the NP-complete
		not-all-equal positive 3-SAT problem (also known as $2$-colouring of
		$3$-uniform hypergraphs).
		This special variant of 3-SAT requires every clause to have at least one true and one
		false literal, and there are no negations allowed.
		We again follow the same general scheme as for $c=1$, but this time 
		the main focus will be on implementing the guard-fix gadget.
		Let our guard colours be red and blue (then every vertex must see exactly
		one red guard, or exactly one blue guard).
		
		The left and right walls of the schematic frame from
		Figure~\ref{fig:hardness-rough} are constructed as shown in
		Figure~\ref{fig:v2v-2colourF}.
		In the construction, we use four copies of the bowl
		shape from Figure~\ref{fig:v2v-2colour}, and we adopt the terminology
		of gluing the bowls and of the {\em doors} from the proof of
		Proposition~\ref{pro:v2v-3colours}.
		For simplicity, while keeping in mind that the door of each bowl is a narrow
		passage formed by a pair of vertices, we denote each door by a single letter
		$d_i$, $i=1,2,3,4$ (as other vertices).
		
		This guard-fix gadget is constructed such that both $d_1,d_2$ see all
		the $8$ vertices of the ``pocket'' $A$ on the left, but neither of $a_1,a_2$ does so.
		We recall the following property from the proof of Proposition~\ref{pro:v2v-3colours}:
		\begin{itemize}
			\item[(i)] In any conflict-free $2$-colouring of the bowl there is a guard
			(or two) placed at the door.
		\end{itemize}
		Consequently, the guards placed at $d_1$ and $d_2$ must be of different
		colours (red and blue); otherwise, say for two red guards at $d_1,d_2$, each
		of the $8$ vertices in $A$ would have to be guarded by a blue guard
		within $A\cup\{a_1,a_2\}$ which is not possible.
		(Though, we have not yet excluded the case that, say, $d_1$ would have 
		a red guard and $d_2$ a blue and a red guards.)
		
		In the next step, we note that $d_3,d_4$ see all $8$ vertices of the 
		``pocket'' $B$, but none of $b_1,b_2,c_3,c_4$ does so.
		So, by analogous arguments, the guards placed at $d_3,d_4$ must be of
		different colours (red and blue).
		We remark that this does not necessarily cause a conflict with the guards
		at $d_1,d_2$ since the visibility between $b_1,d_3$ and between $b_2,d_4$
		is blocked by $c_3$ and $c_4$, respectively.
		However, the vertices $b_1$ and $b_2$ are now ``exhausted'' in the sense
		that one sees (at least) one red and two blue guards, and the other one blue
		and two red guards.
		Altogether, this implies that
		\begin{itemize}\item[(ii)]
			there is exactly one red and one blue guard
			among $d_1,d_2$ and the same holds among $d_3,d_4$,
			and no other vertex visible from both $b_1$ and $b_2$ can have a guard.
		\end{itemize}
		
        \begin{figure}[tbp]
                \def\bowlsh#1#2{%
                        \begin{scope}[shift={#1}, rotate=#2, yscale=0.25, xscale=0.1]
                                \begin{scope}[shift={(0,0.3)}]
                                        \draw[color=blue, thin, fill=blue!15!white] (-0.06,-0.3) -- (-1.9,1) -- (-2.1,2) -- (-1.1,2)
                                        -- (-0.3,1) -- (0.3,1) -- (1.1,2) -- (2.1,2) -- (1.9,1) -- (0.06,-0.3);
                                \end{scope}
                        \end{scope}
                }
                \centering
                \begin{tikzpicture}[yscale=1.8,xscale=1.5]
                \tikzstyle{every node}=[draw, shape=circle, minimum size=3pt,
                inner sep=0pt, fill=gray];
                \tikzstyle{every path}=[thick];
                \node (a) at (0,0) {};
                \node (b) at (8,0) {};
                \coordinate (b1) at (8,1) {};
                \node[color=blue,fill=blue!30!white, minimum size=4pt,
                label=below:$d_4$] (b2) at (8.5,1.02) {};
                \node[color=blue,fill=blue!30!white, minimum size=4pt,
                label=above:$d_3$] (b3) at (8.5,1.18) {};
                \coordinate (b4) at (8,1.2) {};
                \node[color=blue,fill=blue!30!white, minimum size=4pt,
                label=below:$~~d_2\!\!\!\!$] (b5) at (8,2.02) {};
                \node[color=blue,fill=blue!30!white, minimum size=4pt,
                label=above:$~~d_1\!\!\!\!$] (b6) at (8,2.18) {};
                \node (e) at (8,3) {};
                \node[label=$c_1$] (h) at (0,3) {};
                \node[label=$c_2$] (h2) at (0.5,2.85) {};
                \node[label=$c_3$] (h3) at (1,2.74) {};
                \node[label=$c_4$] (h4) at (1.5,2.65) {};
                \node[label=$\dots$] (h5) at (2,2.58) {};
                \coordinate[label=below:$\!\!b_2~~$] (f1) at (0,1) {};
                \coordinate (f2) at (-0.8,1.05) {};
                \coordinate (f3) at (-0.8,1.15) {};
                \coordinate[label=$\!\!b_1~~$] (f4) at (0,1.2) {};
                \coordinate[label=below:$\!\!a_2~~$] (g1) at (0,2.0) {};
                \coordinate (g2) at (-0.8,2.05) {};
                \coordinate (g3) at (-0.8,2.15) {};
                \coordinate[label=$\!\!a_1~~$] (g4) at (0,2.2) {};
                \draw[dashed, color=gray] (a) arc (98:82:28.6) -- (b) ;
                \draw[dashed, color=gray] (h) arc (-105:-75:15.2) -- (e) ;
                \draw (a) -- (f1) -- (f2) -- (f3) -- (f4)
                -- (g1) -- (g2) -- (g3) -- (g4) -- (h);
                \draw (b) -- (b1) -- (b2) -- (b3) -- (b4) --(b5) -- (b6) -- (e);
                \draw (h) -- (h2) -- (h3) -- (h4) -- (h5) -- (2.2,2.56);
                \draw[dotted, color=blue] (b5) -- (2,2.58);
                \draw[dotted, color=blue] (b6) -- (a);
                \draw[dotted, color=gray] (f4) -- (8,2.98);
                
                \node (u1) at (4.78,2.49) {};
                \node[fill=white, label=left:$x_2$] (u2) at (4.95,3.01) {};
                \node (u3) at (5.02,2.5) {};
                \bowlsh{(u2)}{0}
                \draw[dotted, color=green] (u1) -- (4.2,0.1);
                \draw[dotted, color=green] (u3) -- (5.2,0.1);
                \node (w1) at (3.95,2.466) {};
                \node[fill=white, label=left:$x_1$] (w2) at (4.3,3) {};
                \node (w3) at (4.25,2.47) {};
                \bowlsh{(w2)}{0}
                \draw[dotted, color=red] (w1) -- (2.45,0.1);
                \draw[dotted, color=red] (w3) -- (3.85,0.1);
                \node (v1) at (5.5,2.54) {};
                \node[fill=white, label=left:$x_3$, label=right:$~~\dots$] (v2) at (5.5,3.03) {};
                \node (v3) at (5.65,2.56) {};
                \bowlsh{(v2)}{0}
                \draw[dotted, color=red] (v2) -- (5.5,0.08);
                \draw[dotted, color=red] (v2) -- (6.4,0.07);
                \draw (3.7,2.47) -- (w1) -- (w2) -- (w3) -- (u1) -- (u2) 
                --(u3) --(v1) --(v2) --(v3) -- (5.9,2.59);
                
                \node (l1) at (2.8,0.28) {};
                \node (l2) at (3.23,-0.4) {};
                \node (l3) at (2.92,0.285) {};
                \draw[dotted, color=green] (l2) -- (h4);
                \node (ll1) at (3.2,0.3) {};
                \node (ll2) at (4,-0.3) {};
                \node (ll3) at (3.35,0.3) {};
                \draw[dotted, color=green] (ll2) -- (h);
                \draw (2.6,0.27) -- (l1) -- (l2) -- (l3) -- (ll1) -- (ll2) -- (ll3)
                -- (3.6,0.3);
                \node[label=240:$l_1$] (lm1) at (4.7,0.28) {};
                \node[label=right:$l_2$] (lm2) at (5.7,-0.3) {};
                \node[label=60:$l_3$] (lm3) at (4.9,0.278) {};
                \draw[dotted, color=red] (lm2) -- (h3);
                \draw (lm1) -- (lm2) -- (lm3);
                \end{tikzpicture}
                
                \caption{Placement of the variable/reflection/clause gadgets for conflict-free $2$-colouring:
                        the clause gadgets are the single vertices $c_1,c_2,\ldots$,
                        the variable gadgets are formed by copies of the bowl at $x_1,x_2,\ldots$,
                        and the reflection gadgets are simply triples of vertices as
                        $l_1,l_2,l_3$ at the bottom line.
                        In this example, the value (colour red or blue) of the variable
                        $x_1$ is reflected towards clauses $c_1$ and $c_4$ (which are
                        thus assumed to contain literal~$x_1$), and the value of
                        $x_2$ is reflected towards~$c_3$.}
                \label{fig:v2v-2colourG}
        \end{figure}
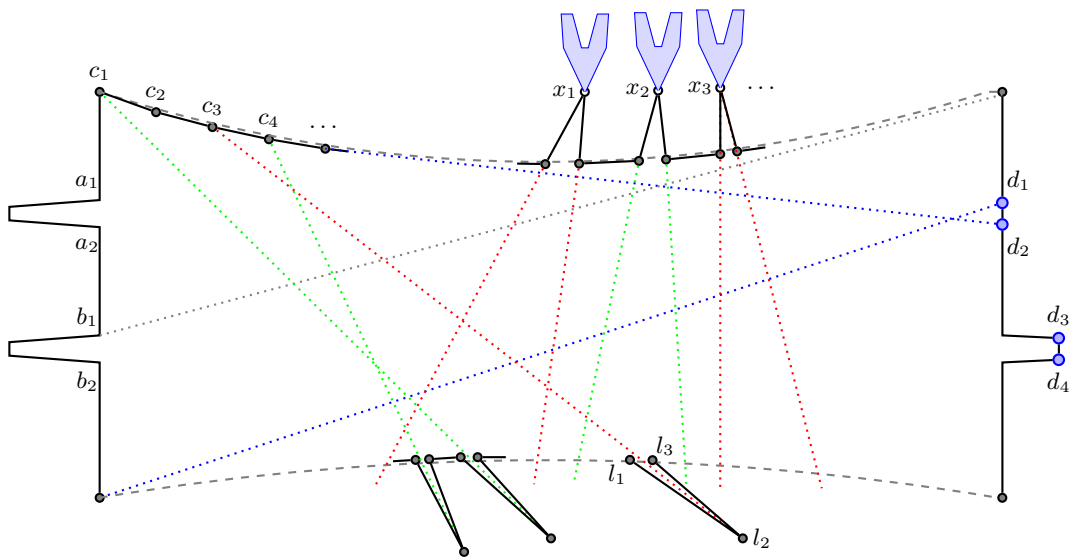
		
		The rest of the construction is, within the frame constructed above
		(Figure~\ref{fig:v2v-2colourF}), already quite easy.
		Let $\Phi$ be a given 3-SAT formula without negations.
		See Figure~\ref{fig:v2v-2colourG}.
		\begin{itemize}
			\item We again represent each clause of $\Phi$ by a single vertex on a
			concave chain on the top of our frame.
			This chain of clause vertices is ``slightly hidden'' in a sense that it is
			not visible from $d_1$ or $d_2$, but it is all visible from $b_1$ and~$b_2$.
			\item Each variable $x_i$ of $\Phi$ is represented by a copy of the bowl,
			also placed on the top of the frame (but separate from the
			section of clause vertices).
			As before, the visible angles of the variable gadgets are adjusted 
			so that they do not overlap on the bottom line of the frame.
			\item Each literal $\ell=x_i$ is represented by a triple of vertices
			as $l_1,l_2,l_3$ in Figure~\ref{fig:v2v-2colourG}, such that $l_1,l_3$ are
			in the visible angle of the variable-$x_i$ gadget, while $l_2$ is ``deeply
			hidden'' so that $l_2$ sees only the clause vertex $c_j$ that $\ell$ belongs to.
		\end{itemize}
		
		The final argument is analogous to the $c=1$ case.
		Assume we have a not-all-equal assignment of~$\Phi$.
		We put red guards to $d_1$ and $d_3$ and blue guards to $d_2$ and~$d_4$.
		For each variable $x_i$, we put one guard to $x_1$ of blue colour if $x_i$
		is true and of red colour if $x_i$ is false.
		Whenever $\ell=x_i$ is a literal represented by the triple $l_1,l_2,l_3$,
		we put to $l_2$ a guard of the same colour as of the guard at $x_i$.
		Then every clause $c_j=(x_a\vee x_b\vee x_c)$ will see the colours of guards
		at $x_a,x_b,x_c$, and since the values assigned to $x_a,x_b,x_c$ are not all
		the same, one of the colours is unique to guard~$c_j$.
		We have got a conflict-free $2$-colouring.
		
		On the other hand, assume a conflict-free $2$-colouring.
		By (ii), there are no guards at the clause vertices $c_1,c_2,\ldots$,
		and so those vertices can be only guarded from the reflection gadgets.
		Assume a reflection gadget of a literal $\ell=x_i$.
		Then, again by (ii), the vertices $l_1,l_3$ of this gadget have no guards,
		and they see a red and a blue guard from $d_1,d_2$.
		On the other hand, $l_2$ cannot see any other guard except one placed at
		$l_2$, and so there has to be a guard at~$l_2$.
		If, moreover, $l_1,l_3$ see a red (say) guard placed at $x_i$
		(and $x_i$ as the door of a glued bowl there must have a guard by (i)\,),
		then the guard at $l_2$ must also be red (or $l_1$ would have a conflict).
		Consequently, every clause $c_j=(x_a\vee x_b\vee x_c)$ sees the colours of
		the guards placed at $x_a$, $a_b$ and $x_c$, and since the colouring is
		conflict-free, there have to be both colours visible (one red plus two blue,
		or one blue plus two red).
		From this we can read a valid not-all-equal variable assignment of~$\Phi$.
	\end{proof}

	\section{Concluding remarks}
	
	We have designed an algorithm for producing a V2P guarding of funnels that is 
	optimal in the number of guards.
	We have also designed an algorithm for a V2P conflict-free chromatic guarding 
	for funnels, which gives only a small additive error
	with respect to the minimum number of colours required. 
	We believe that the latter can be strengthened to an exact solution
	by sharpening the arguments involved (though, it would likely not be easy).
	
	For V2P conflict-free chromatic guarding simple polygons, we have given a simple
	efficient upper bound of $O(\log^2 n)$, first for the special case of weak-visibility 
	polygons and then, as a corollary, for all simple polygons.
	This still leaves room for improvement down to $O(\log n)$, which is
	the worst-case scenario already for funnels and which would match
	the previous P2P upper bound for simple polygons of \cite{Bartschi-2014}.
	We believe that such an improvement is achievable by a more careful
	distribution of the colours sets (and re-use of colours)
	to the max funnels in Algorithm~\ref{alg:weakv2pcolouringB}.

	For the V2V guarding case, it is easy to see that there is a natural upper bound 
	by the domination number of the visibility graph.
	Also, in the V2V case for funnels, our Algorithm~\ref{alg:tightpath} can be used 
	to decide whether a given funnel can be guarded by only one colour,
	whereas a trivial V2V conflict-free chromatic guarding with two colours 
	exists for any funnel (up to some technical details, every third
	vertex of the left chain is red, and every third vertex of the right chain is blue).

	For general simple polygons, we have shown that the problem of V2V conflict-free
	chromatic guarding is NP-complete when the number of colours is one or two.  
	We have also given a simple example where the polygon requires three
	colours, and conjectured that no more colours are needed
	no matter how convoluted the polygon might be.

	While we have shown that the problem of V2V conflict-free
        chromatic guarding is NP-complete already with one or two colours,
	this does not readily imply any hardness results for the V2P and P2P cases.
	Hence, the complexity of an exact solution in the latter two cases remains open.
	
	\smallskip
	To summarise, we propose the following open problems for future research:
	\begin{itemize}\vspace*{-1ex}%
		\item Improve Algorithm~\ref{alg:apxcfreefunnel} to an exact algorithm for chromatic guarding a funnel.
		\item Improve the upper bound in Theorem~\ref{thm:twoApprox} to $O(\log n)$.
 		\item Prove Conjecture~\ref{cj:v2v-upper3} or, at least, prove any constant upper bound on the V2V conflict-free chromatic guarding
 		of simple polygons.
	\end{itemize}

\bibliographystyle{spmpsci}      % mathematics and physical sciences
\bibliography{jco}   % name your BibTeX data base

\end{document}